\documentclass[11pt,journal,onecolumn,letterpaper]{IEEEtran}
\usepackage[utf8]{inputenc}
\usepackage[T1]{fontenc}
\usepackage[hidelinks]{hyperref}
\usepackage{ifthen}
\usepackage{cite}
\usepackage{xcolor}
\usepackage{amsthm,amsfonts,amssymb,arydshln}
\usepackage{mathrsfs}
\usepackage{graphicx}
\usepackage{caption}
\usepackage{subcaption}
\usepackage[cmex10]{amsmath}
\usepackage{enumitem}
\usepackage{float,needspace}
\usepackage{microtype}
\usepackage{tikz}
\usetikzlibrary{decorations.pathreplacing,decorations.markings,decorations.pathmorphing,calc,arrows.meta}

\usetikzlibrary{arrows.meta,positioning,calc}

\usepackage{amsmath, amssymb, bm}
\usepackage{mathtools}

\newcommand{\ve}[1]{\bm{#1}}      %
\newcommand{\m}[1]{\mathbf{#1}}	 %
\newcommand{\E}[1]{\mathbb{E}\left[ #1 \right]}
\newcommand{\tr}[1]{\text{tr}\left( #1 \right)}
\newcommand{\rank}[1]{\text{rank}\left( #1 \right)}

\newcommand{\new}[1]{{\color{black}#1}}

\theoremstyle{definition}
\newtheorem{theorem}{Theorem}

\newtheorem{lemma}{Lemma}
\newtheorem{corollary}{Corollary}
\newtheorem{definition}{Definition}

\newtheorem{remark}{Remark}

\definecolor{softblack}{HTML}{222222}
\definecolor{softgray}{HTML}{333333}
\definecolor{deepblue}{HTML}{1F4E79}
\definecolor{deepgreen}{HTML}{1B5E20}
\definecolor{deeprose}{HTML}{7F1D1D}

\allowdisplaybreaks

\begin{document}

\title{Degrees of Freedom of Over-the-Air Computation over a MIMO Gaussian Network with Two Transmitters and Two Receivers}

\author{Yong~Dong, Hua~Sun,~\IEEEmembership{Senior Member,~IEEE,} and Syed~A.~Jafar,~\IEEEmembership{Fellow, IEEE.}

\thanks{
Yong Dong and Syed A. Jafar are with the Department of Electrical Engineering and Computer Science, University of California at Irvine, Irvine, CA 92697 USA (e-mail:yongd3@uci.edu; syed@uci.edu). Hua Sun is with the Department of Electrical Engineering, University of North Texas, Denton, TX 76203 USA (e-mail: hua.sun@unt.edu).}

\thanks{\copyright~2026 IEEE. All rights reserved, including rights for text and data mining, and training of artificial intelligence and similar technologies. Personal use is permitted, but republication/redistribution requires IEEE permission. See \href{https://www.ieee.org/publications/rights/index.html}{IEEE Publications Rights} for more information.
Published in \emph{IEEE Transactions on Communications}. DOI: \href{https://doi.org/10.1109/TCOMM.2026.3745167}{10.1109/TCOMM.2026.3745167}.}}

\maketitle

\begin{abstract}
The fundamental limits of over-the-air computation (AirComp) are explored in a two-transmitter, two-receiver multiple-input multiple-output (MIMO) Gaussian network, where both receivers demand the same aggregation of source symbols originating at the two transmitters. An AirComp degrees of freedom (DoF) metric is defined under a noise-limited mean-squared error constraint. For a generic MIMO setting where the two transmitters are equipped with $M_1, M_2$ antennas, and the two receivers with $N_1, N_2$ antennas, the AirComp DoF value is shown to be almost surely equal to $\min\{M_1,M_2,N_1,N_2,(1/3)\max\{M_1+M_2,N_1+N_2\}\}$. For single-input single-output (SISO) settings\new{,} results are extended beyond generic channels to arbitrary channel realizations \new{ satisfying explicit channel conditions}. For finite signal-to-noise ratio (SNR) settings, an iterative alternating optimization algorithm is explored.
\end{abstract}


\section{Introduction}

\subsection{\texorpdfstring{\new{Over-the-Air Computation}}{Over-the-Air Computation}}

\new{Over-the-air computation (AirComp) is a physical-layer technique that exploits the superposition property of multiple access channels (MACs) to perform computations directly during transmission, without first decoding all users' messages. It is particularly suitable for nomographic functions of the users' signals, with sum computation as a canonical example \cite{goldenbaumNomographicFunctionsEfficient2015,goldenbaumHarnessingInterference2013}. AirComp schemes differ in the alphabet of the source messages, in the alphabet of the channel inputs, and in the performance metric they adopt. We review the main lines of work along these axes.}

\new{In the line of work closest to this paper, the source symbols are continuous-valued and are transmitted directly as scaled channel inputs, and the accuracy of the estimated function is measured by the mean-squared error (MSE). Within this line, channel-inversion, zero-forcing, and uniform-forcing designs perform exact function alignment by forcing the effective source coefficients to their desired values, leaving additive noise as the only error \cite{zhuMIMOOvertheAirComputation2018,chenUniformForcingTransceiverDesign2018}. Optimized power control and MMSE-type designs instead permit residual coefficient misalignment and trade the resulting error against noise amplification \cite{caoOptimizedPowerControl2020a,zangOvertheAirComputationSystems2020,liuOvertheAirComputationSystems2020a}. This line includes the over-the-air aggregation of local gradients or model updates in federated edge learning \cite{yangFederatedLearningOverTheAir2020,zhuBroadbandAnalogAggregation2020}, possibly aided by reconfigurable intelligent surfaces \cite{huRISAssistedOvertheAirFederated2022}, as well as blind federated edge learning, which removes the transmitter-side channel state information (CSI) assumption altogether and aligns the superposition only asymptotically in the number of receive antennas \cite{amiriBlindFederatedEdge2021}.}

\new{When the source messages are instead drawn from discrete alphabets, computing over the MAC becomes a coding problem. For discrete linear MACs whose algebraic structure matches the desired function, linear codes over finite fields enable the receiver to decode the function of the messages, rather than all individual messages \cite{nazerComputationMultipleAccessChannels2007}. Compute-and-forward extends this idea to Gaussian channels: finite-field messages are mapped to lattice codewords, and the receiver decodes integer combinations of them \cite{nazerComputeandForwardHarnessingInterference2011}. Subsequent work studied, from a computation-rate viewpoint, when the MAC should be used for communicating messages versus computing functions \cite{chenCommunicatingComputingMAC2019} and how opportunistic sensor scheduling can sustain the computation rate in large networks \cite{jeonOpportunisticFunctionComputation2016}.}

\new{A related family also operates over Gaussian channels, whose inputs may take arbitrary continuous values; the schemes themselves, however, deliberately confine the transmitted signals to discrete sets---finite modulation constellations or lattice points---either because the source data is itself digital, or by quantizing continuous-valued data as part of the transmission scheme: one-bit over-the-air aggregation quantizes each continuous-valued gradient entry to a single bit and modulates it onto a fixed constellation for federated edge learning \cite{zhuOneBitOvertheAirAggregation2021}; the \emph{ChannelComp} framework designs constellations under which the desired function of the digital source data can be recovered from the superimposed channel output \cite{razavikiaChannelCompGeneralMethod2024}, and its sum-type specialization \emph{SumComp} performs the computation by coding over the ring of integers \cite{razavikiaSumCompCodingDigital2025}; lattice-coded federated learning quantizes continuous-valued model updates to lattice points and decodes integer combinations of the resulting codewords \cite{azimiabarghouyiComputeUpdateFederated2024}; and out-of-air computation dither-quantizes continuous-valued data onto nested lattice constellations, whose layered transmission acts as a joint source--channel code from which the receiver extracts the desired function \cite{outOfAirComputation}.}

\new{Discrete channel inputs are not, however, a prerequisite for coding: channel-coded AirComp applies an identical linear encoding matrix over the complex field to the continuous-valued source symbols, so that both the messages and the channel inputs remain continuous while the code structure is preserved under summation \cite{wengChannelCodedOverAirComputation2026}. Across these settings, the performance criteria include computation rate, MSE and distortion, reliability, learning convergence, CSI robustness, and transceiver compatibility. In this work, the source symbols are continuous-valued and the channel inputs are not restricted to any discrete set, subject only to an average power constraint; we measure high-SNR computation efficiency by a degrees of freedom (DoF) metric under an MSE accuracy constraint.}

\new{A broad overview of the main AirComp principles and of the literature available at the time is given in the survey \cite{sahinSurveyOvertheAirComputation2023}. The works reviewed above, including the more recent developments, deliver each computed function to a single fusion receiver; their underlying model therefore remains a many-to-one computation channel. They do not address the question studied here: how the AirComp DoF changes when multiple receivers demand the same aggregate.}

\subsection{\texorpdfstring{\new{A Degrees of Freedom Perspective}}{A Degrees of Freedom Perspective}}

\new{Consider a receiver that wishes to compute the sum $s_1+\cdots+s_K$ of the source symbols held by $K$ transmitters. AirComp can compute this sum in one channel use, irrespective of $K$. Thus, in terms of degrees of freedom, aggregating $K$ sources costs no more than delivering one symbol. The underlying mechanism is the \emph{ideal sum channel}. To illustrate, consider the two-user single-input single-output (SISO) MAC $y=h_1x_1+h_2x_2+z$. For notational simplicity, we omit the common power-normalization factor, which is fixed for a given channel realization and does not affect the high-SNR MSE scaling. Channel inversion $x_i=s_i/h_i$ then yields $y=s_1+s_2+z$. Hence, $\hat f=y$ is an unbiased estimate of the desired sum with $\mathrm{MSE}=O(1/\rho)$. This is the same accuracy scaling obtained when the sum itself is observed through a noisy channel. Matching the effective channel to the desired function in this manner underlies much of the AirComp literature based on direct transmission of continuous-valued symbols. Since each use of the SISO MAC yields one ideal sum, its AirComp DoF is $1$ (Lemma~\ref{lemma:SISOMAC}). For a multiple-input multiple-output (MIMO) MAC with $M_1,M_2$ transmit and $N$ receive antennas, the AirComp DoF generalizes to $\min\{M_1,M_2,N\}$ (Remark~\ref{rem:mimo_mac} in Appendix~\ref{appendix_sisomac}).}

\subsection{\texorpdfstring{\new{AirComp for Two Receivers}}{AirComp for Two Receivers}}

\new{The problem becomes more involved when a second receiver requests the same sum, as illustrated in Figure~\ref{fig:1_b}. Consider the SISO network $y_1=h_{11}x_1+h_{12}x_2+z_1$ and $y_2=h_{21}x_1+h_{22}x_2+z_2$, where both receivers demand $f=s_1+s_2$. Inverting the channels toward Rx-$1$, i.e., $x_1=s_1/h_{11}$ and $x_2=s_2/h_{12}$, creates the ideal sum channel at Rx-$1$. However, Rx-$2$ then observes $y_2=(h_{21}/h_{11})s_1+(h_{22}/h_{12})s_2+z_2$, which is generically not the desired sum. Equivalently, transmit-side scaling alone cannot simultaneously align the effective coefficient vectors at both receivers with the desired direction $(1,1)$. Two time-sharing baselines are available: transmitting $s_1$ and $s_2$ in separate channel uses and summing the estimates at each receiver, or inverting toward one receiver per channel use. Either way, two channel uses are spent per sum---an AirComp DoF of $1/2$, halving the advantage of AirComp. These furnish the natural time-sharing baselines evaluated in Section~\ref{sec:numerical}, alongside a compound-MAC source-recovery benchmark.}

\new{Time sharing, however, does not fully utilize the broadcast nature of the wireless medium. While one receiver is being served, the other also obtains a side observation whose effective coefficient vector is generally not aligned with the desired direction $(1,1)$. This motivates joint processing across transmitters, receivers, and channel uses to achieve simultaneous function alignment at both receivers. We therefore ask how many independent ideal-sum dimensions per channel use the two-transmitter, two-receiver network can support. To formalize this question, we define the AirComp DoF as the number of computations per channel use whose estimation error can be maintained at the channel-noise scale, i.e., $\mathrm{MSE}=O(1/\rho)$ as $\rho\to\infty$. This is the asymptotic counterpart of the unbiased, noise-limited estimation that channel inversion achieves in the single-receiver MAC.}

\subsection{\texorpdfstring{\new{Summary of Results}}{Summary of Results}}

\new{For the SISO two-receiver network, we show that the AirComp DoF is $2/3$ whenever the cross ratio $h(\nu)=\frac{h_{11}(\nu)h_{22}(\nu)}{h_{12}(\nu)h_{21}(\nu)}$ has a nonzero imaginary part at every channel use (Theorem~\ref{thm:1}). Operationally, three channel uses create two ideal sum channels at each receiver, and no scheme can do better. This exceeds the time-sharing value of $1/2$ but remains below the single-receiver AirComp DoF of $1$. In particular, the generic SISO network has AirComp DoF $2/3$ almost surely (Corollary~\ref{cor:siso_generic}). As a practical consequence, the single-receiver function-alignment techniques built upon the ideal sum channel \cite{chenUniformForcingTransceiverDesign2018, zhuBroadbandAnalogAggregation2020, yangFederatedLearningOverTheAir2020, zhuMIMOOvertheAirComputation2018} can be operated at both receivers at $2/3$ of the single-receiver DoF, rather than the $1/2$ attained by time sharing. The problem is further enriched when the nodes are equipped with multiple antennas. With $M_1,M_2$ transmit and $N_1,N_2$ receive antennas, we characterize how the spatial resources should be allocated to maximize the AirComp DoF. For generic MIMO channels, we obtain the complete characterization $\min\{M_1,M_2,N_1,N_2,\frac{1}{3}\max\{M_1+M_2,N_1+N_2\}\}$ (Theorem~\ref{thm:4}). At finite SNR, because transmit-side inversion cannot align the desired coefficients at both receivers simultaneously, the transmit precoders and receive combiners must be designed jointly. We develop an alternating-optimization transceiver design and compare its observed high-SNR behavior with the theoretical thresholds. Across the simulated cases, the numerical results are consistent with the AirComp DoF characterization, whereas the time-sharing baselines exhibit their spatial limitation (Section~\ref{sec:numerical}).}

\begin{figure}[!t]
    \centering

    \tikzset{
        node distance=0.5cm and 0.5cm,
        font=\scriptsize,
        block/.style={
            draw, rectangle,
            minimum width=0.75cm, minimum height=0.45cm,
            inner sep=2pt, align=center, fill=white
        },
        antennas/.style={
            append after command={
                \pgfextra{
                    \foreach \x in {-0.2, 0, 0.2} {
                        \draw ([xshift=\x cm]\tikzlastnode.north) -- ++(0,0.15)
                              -- ++(-0.05,0.05);
                        \draw ([xshift=\x cm]\tikzlastnode.north) ++(0,0.15)
                              -- ++(0.05,0.05);
                    }
                }
            }
        },
        sig/.style={-Stealth, thick},
        int/.style={-Stealth, thick, dashed, red!80!black},
        cell/.style={draw, circle, minimum size=3cm, thick}
    }

    \begin{subfigure}[t]{0.39\linewidth}
        \centering
        \resizebox{\linewidth}{!}{%
        \begin{tikzpicture}
            \coordinate (c1) at (0,0);
            \coordinate (c2) at (2.8,0);

            \node[cell] at (c1) {};
            \node[cell] at (c2) {};

            \node[block, antennas] (rx1) at (0, 0.7) {Rx-$1$};
            \node[block, antennas] (tx1) at (-0.6, -0.7) {Tx-$1$};
            \node[block, antennas] (tx2) at (0.6, -0.7) {Tx-$2$};

            \node[block, antennas] (rx2) at (2.8, 0.7) {Rx-$2$};
            \node[block, antennas] (tx3) at (2.2, -0.7) {Tx-$3$};
            \node[block, antennas] (tx4) at (3.4, -0.7) {Tx-$4$};

            \draw[sig] (tx1) -- (rx1);
            \draw[sig] (tx2) -- (rx1);
            \draw[sig] (tx3) -- (rx2);
            \draw[sig] (tx4) -- (rx2);

            \draw[int] (tx1) -- (rx2.west);
            \draw[int] (tx2) -- (rx2.south west);
            \draw[int] (tx3) -- (rx1.south east);
            \draw[int] (tx4) -- (rx1.east);
        \end{tikzpicture}%
        }
        \caption{}
        \label{fig:1_a}
    \end{subfigure}
    \hspace{0.06\linewidth}
    \begin{subfigure}[t]{0.195\linewidth}
        \centering
        \resizebox{\linewidth}{!}{%
        \begin{tikzpicture}
            \coordinate (c3) at (0,0);

            \node[cell] at (c3) {};

            \node[block, antennas] (rx1b) at (-0.6, 0.6) {Rx-$1$};
            \node[block, antennas] (rx2b) at (0.6, 0.6) {Rx-$2$};

            \node[block, antennas] (tx1b) at (-0.6, -0.8) {Tx-$1$};
            \node[block, antennas] (tx2b) at (0.6, -0.8) {Tx-$2$};

            \draw[sig] (tx1b) -- (rx1b);
            \draw[sig] (tx1b) -- (rx2b);
            \draw[sig] (tx2b) -- (rx1b);
            \draw[sig] (tx2b) -- (rx2b);
        \end{tikzpicture}%
        }
        \caption{}
        \label{fig:1_b}
    \end{subfigure}

    \caption{AirComp models. (a) Simultaneous AirComp at two receivers, each with its own pair of transmitters \cite{caoCooperativeInterferenceManagement2021, liMultiCellOvertheAirComputation2023, lanSimultaneousSignalandInterferenceAlignment2020, liChannelReconfigurationDistributed2024}. (b) Simultaneous AirComp at two receivers from the same two transmitters (this work).}

    \label{fig:aircomp_model}
\end{figure}

\subsection{\texorpdfstring{\new{Related Work on Multiple-Receiver AirComp}}{Related Work on Multiple-Receiver AirComp}}

Simultaneous AirComp at multiple receivers, each with its own set of desired transmitters (see Figure~\ref{fig:1_a}), has been studied in \cite{caoCooperativeInterferenceManagement2021, liMultiCellOvertheAirComputation2023, lanSimultaneousSignalandInterferenceAlignment2020, liChannelReconfigurationDistributed2024}. Inter-cell interference is managed through cooperative power control over the per-cell MSEs~\cite{caoCooperativeInterferenceManagement2021,liMultiCellOvertheAirComputation2023} or through signal-and-interference alignment, under which two cells with $2N$ antennas per node support $N$ simultaneous AirComp tasks each~\cite{lanSimultaneousSignalandInterferenceAlignment2020}, with distributed designs for general multi-cell networks~\cite{liChannelReconfigurationDistributed2024}.

\new{In contrast, this paper explores the \emph{single-cell, multi-receiver} scenario illustrated in Figure~\ref{fig:1_b}.} In the multi-cell case, receivers seek different local aggregates \new{ and must suppress interference from transmitters serving other cells. Thus, the multi-cell setting is closer to a conventional interference-alignment problem: the interfering signals are aligned into a common subspace that is linearly independent of the desired-signal subspace.} Here both receivers seek the same aggregate of the same transmitters. \new{Neither source is interference; instead, the challenge is simultaneous function alignment---the two desired sources must have equal effective coefficients at each receiver. A related objective, function or computation alignment, has been studied for multicast function computation \cite{suhComputationMulticastNetworks2016,suhNetworkDecompositionFunction2013} and its interactive extension \cite{suhInteractiveFunctionComputation2013} under the noiseless finite-field ADT model \cite{Avestimehr_Diggavi_Tse}. Its exact deterministic-network decomposition inspires our construction, but a generic Gaussian MIMO channel does not decompose physically. Instead, we establish virtual packing by exhibiting one channel realization at which the rational polynomials certifying full rank are nonzero; they are then nonzero almost surely for generic channels (Section~\ref{proof:MIMOach}). The converse is likewise adapted to noisy computation through MSE--equivocation and invertibility bounds.}

\subsection{Notation}
For a matrix $\m{A}$, $\m{A}[m,n]$ denotes its $(m,n)$-th entry. A sequence $(X(1), X(2), \ldots, X(N))$ is denoted by $X^N$.
The transpose, conjugate transpose (Hermitian), rank, and trace of a matrix $\m{A}$ are denoted by $\m{A}^T$, $\m{A}^\dagger$, $\rank{\m{A}}$, and $\tr{\m{A}}$, respectively. The Frobenius norm of a matrix $\m{A}$ is denoted by $\lVert \m{A} \rVert_\mathrm{F}$.
The set $\{1, 2, \ldots, n\}$ is denoted by $[n]$\new{; more generally, for integers $a \le b$, $[a:b] \triangleq \{a, a+1, \ldots, b\}$. The set of positive integers is denoted by $\mathbb{N}$.} For the relations $=$, $\ge$, and $\le$, the notation $a \overset{\mathrm{a.s.}}{=} b$ (and similarly for $\overset{\mathrm{a.s.}}{\ge}$ and $\overset{\mathrm{a.s.}}{\le}$) indicates that the respective relation between $a$ and $b$ holds almost surely, i.e., with probability one. We say that $f(x)=O(g(x))$ as $x\rightarrow 0$, if there exist positive constants $c,\delta$ such that $|f(x)|\leq c|g(x)|$ for  $0<|x|<\delta$. We say $f(y)=O(g(y))$ as $y\rightarrow \infty$, or simply $f(y)=O(g(y))$ if there exist positive constants $c,\delta$ such that $|f(y)|\leq c|g(y)|$ for  $y>\delta$.  Finally, $f(x)=o(g(x))$ as $x\rightarrow \infty$, if for every $c>0$ there exists $\delta(c)>0$ such that $|f(x)|\leq c|g(x)|$ for  $x>\delta(c)$.

\section{Problem Formulation}\label{sec:problem_formulation}

\begin{figure}[!t]
    \centering
    \resizebox{0.62\linewidth}{!}{%
    \begin{tikzpicture}[
        thick,
        every node/.style = {font=\small},
        block/.style = {draw, rectangle, minimum width=18pt, minimum height=38pt},
        ant/.style   = {circle, draw, inner sep=1pt}
    ]

    \node[block] (X1) at (0,  1.2) {};
    \node[above=2pt of X1] {Tx-$1$};

    \node[block] (X2) at (0, -1.2) {};
    \node[above=2pt of X2] {Tx-$2$};

    \node[block] (Y1) at (6,  1.2) {};
    \node[above=2pt of Y1] {Rx-$1$};

    \node[block] (Y2) at (6, -1.2) {};
    \node[above=2pt of Y2] {Rx-$2$};

    \foreach \rect in {X1,X2,Y1,Y2}{
        \node[ant] at ($(\rect.center)+(0, 12pt)$) {};
        \node[ant] at ($(\rect.center)+(0, 6pt)$) {};
        \node at  ($(\rect.center)+(0,0pt)$) {$\vdots$};
        \node[ant] at ($(\rect.center)+(0,-12pt)$) {};
    }

    \node[circle, draw, inner sep=1pt, minimum size=12pt] (s1) at (5,  1.2) {$+$};
    \node[circle, draw, inner sep=1pt, minimum size=12pt] (s2) at (5, -1.2) {$+$};

    \node (z1) at (5,  2.2) {$\ve{z}_1$};
    \node (z2) at (5, -0.2) {$\ve{z}_2$};

    \draw[->, >=Latex] (X1.east) -- (s1.west) node[midway, above] {$\m{H}_{11}$};
    \draw[->, >=Latex] (X2.east) -- (s2.west) node[midway, below] {$\m{H}_{22}$};

    \draw[->, >=Latex] (X1.east) -- (s2.west)
          node[midway, below left, yshift=-8pt] {$\m{H}_{12}$};
    \draw[->, >=Latex] (X2.east) -- (s1.west)
          node[midway, above left, yshift= 8pt] {$\m{H}_{21}$};

    \draw[->, >=Latex] (z1) -- (s1);
    \draw[->, >=Latex] (z2) -- (s2);

    \draw[->, >=Latex] (s1.east) -- (Y1.west);
    \draw[->, >=Latex] (s2.east) -- (Y2.west);

    \node (S1) at (-2,  1.2) {$\{s_1(k)\}$};
    \draw[->, >=Latex] (S1.east) -- (X1.west);

    \node (S2) at (-2, -1.2) {$\{s_2(k)\}$};
    \draw[->, >=Latex] (S2.east) -- (X2.west);

    \node (F1) at (8,  1.2) {$\{\hat{f}_1(k)\}$};
    \draw[->, >=Latex] (Y1.east) -- (F1.west);

    \node (F2) at (8, -1.2) {$\{\hat{f}_2(k)\}$};
    \draw[->, >=Latex] (Y2.east) -- (F2.west);

    \end{tikzpicture}%
    }
    \caption{Two-transmitter, two-receiver MIMO Gaussian network.}
    \label{fig:net}
\end{figure}

Consider a memoryless, discrete-time two-transmitter, two-receiver network in the complex domain with additive white Gaussian noise (AWGN). Tx-$1$ and Tx-$2$ have $M_1$ and $M_2$ antennas, respectively, and Rx-$1$ and Rx-$2$ have $N_1$ and $N_2$ antennas, respectively, as shown in Figure~\ref{fig:net}. Over the $\nu$-th channel use, the received signal $\ve{y}_i(\nu) \in \mathbb{C}^{N_i}$ at Rx-$i$, $i\in\{1,2\}$, is given by
\begin{align}
    \ve{y}_i(\nu) &= \m{H}_{i1}(\nu)\ve{x}_1(\nu) + \m{H}_{i2}(\nu)\ve{x}_2(\nu) + \ve{z}_i(\nu), \label{eq:modely1}
\end{align}
where $\ve{x}_1(\nu) \in \mathbb{C}^{M_1}$ and $\ve{x}_2(\nu) \in \mathbb{C}^{M_2}$ are the transmitted signals from Tx-$1$ and Tx-$2$, respectively. To avoid degenerate scenarios, the channel coefficients are assumed to be bounded away from zero and infinity: there exists a finite positive constant $M_o$ such that $1/M_o\leq |\m{H}_{ij}(\nu)[m,n]|\leq M_o$ for all $i,j\in\{1,2\}$, $\nu\in\mathbb{N}$, $m\in[1:N_i]$, and $n\in[1:M_j]$. Tx-$1$ and Tx-$2$ are subject to transmit power constraints, normalized to unity (by a corresponding scaling of the channel coefficients) without loss of generality. For a coding scheme using $N$ channel uses, the power constraint is
\begin{align} \label{eq:P}
  \frac{1}{N}  \sum_{\nu=1}^N \E{\lVert\ve{x}_i(\nu)\rVert^2} \le 1, \; i \in \{1,2\}.
\end{align}
The additive noise vectors are independent across receivers and channel uses, with $\ve{z}_i(\nu)\sim\mathcal{CN}(\mathbf{0},\sigma^2\m{I}_{N_i})$. The signal-to-noise ratio (SNR) is defined as
\begin{align}
\rho \triangleq 1/\sigma^2.\label{def:rho}
\end{align}
Based on \eqref{def:rho}, the high-SNR limit $\lim_{\rho\to\infty}(\cdot)$ is used interchangeably with $\lim_{\sigma^2\to 0}(\cdot)$. Perfect channel state information is assumed at both transmitters (perfect CSIT) and receivers (perfect CSIR). \new{The baseband model further assumes perfect timing, carrier-frequency, and phase synchronization; synchronization errors and hardware-induced amplitude mismatches are not modeled.}

The channel from Tx-$j$ to Rx-$i$ at the $\nu$-th channel use is represented by $\m{H}_{ij}(\nu) \in\mathbb{C}^{N_i\times M_j}$, $i,j\in\{1,2\}$. Define the block matrix
\begin{align} \label{eq:a_H}
    \m{H}(\nu) &\triangleq \begin{bsmallmatrix}\m{H}_{11}(\nu)&\m{H}_{12}(\nu)\\ \m{H}_{21}(\nu)&\m{H}_{22}(\nu)\end{bsmallmatrix},
    \end{align}
for  $\nu \in \mathbb{N}$.  A sequence of random channel matrices is denoted by $\mathcal{H} \triangleq \{\m{H}(\nu)\}_{\nu\in\mathbb{N}}$, whereas  $\mathcal{H}=\mathscr{H}$ denotes a fixed realization.

Various channel models considered in this work are classified as follows. The channel is \emph{time-invariant} if $\m{H}(\nu_1)=\m{H}(\nu_2)$ for all $\nu_1,\nu_2\in\mathbb{N}$. The channel is \emph{time-varying} otherwise. $\mathcal{H}$ is \emph{i.i.d.\ time-varying} if the matrices $\m{H}(\nu)$ are independent and identically distributed across $\nu\in\mathbb{N}$. $\mathcal{H}$ is \emph{generic} if the entries of $\m{H}(\nu)$ are drawn i.i.d.\ from a continuous distribution.

As a compact label, the tuple $(M_1, M_2; N_1, N_2; \mathcal{H})$ denotes the network with two transmitters and two receivers, where $M_i$, $i \in \{1,2\}$, denotes the number of antennas at Tx-$i$, and $N_j$, $j \in \{1,2\}$, denotes the number of antennas at Rx-$j$. For a fixed channel realization $\mathcal{H}=\mathscr{H}$, the network is denoted by $(M_1, M_2; N_1, N_2; \mathscr{H})$, which may be abbreviated as an $(M_1, M_2; N_1, N_2)$ network or simply as $\mathscr{H}$.

Each Tx-$i$, $i \in \{1, 2\}$, obtains an independent data sequence $\{s_i(\kappa) \in \mathbb{C}\}_{\kappa \in \mathbb{N}}$ whose symbols are drawn i.i.d.\ from a complex Gaussian distribution $\mathcal{CN}(0, \sigma^2_s)$. Each Rx-$j$, $j \in \{1, 2\}$, wishes to compute the sum sequence $\{f(\kappa) = s_1(\kappa) + s_2(\kappa)\}_{\kappa \in \mathbb{N}}$. \new{The AirComp DoF characterization also applies to fixed nonzero weighted sums $a_1s_1+a_2s_2$, because Tx-$i$ can encode $a_i s_i$. For a nomographic function $\psi(\sum_i\phi_i(s_i))$, neither the DoF nor the end-to-end MSE follows directly because preprocessing changes the source statistics and nonlinear postprocessing changes error propagation. Our construction remains useful as a physical-layer primitive: it creates parallel ideal sum channels at both receivers for computing $\sum_i\phi_i(s_i)$, after which each receiver applies $\psi$.}

\begin{definition}[Computation Code] \label{def:1}
A $(K, N)$ \emph{computation code} for an $(M_1,M_2;N_1,N_2)$ channel is denoted as $\mathfrak{C}_{K,N}=(\mathcal{E}_1,\mathcal{E}_2,\mathcal{D}_1,\mathcal{D}_2)$ and consists of the following.
\begin{itemize}
    \item
    Encoding functions $\mathcal{E}_i: \mathbb{C}^K\rightarrow \mathbb{C}^{NM_i}$ for $i\in\{1,2\}$ such that $\mathcal{E}_1(s_1^K) = \ve{x}_1^N$, $\mathcal{E}_2(s_2^K) = \ve{x}_2^N$, and \eqref{eq:P} is satisfied.
    \item
    Decoding (estimation) functions $\mathcal{D}_i:\mathbb{C}^{NN_i}\rightarrow\mathbb{C}^{K}$ for $i\in\{1,2\}$ such that $\mathcal{D}_1(\ve{y}_1^N) = \hat{f}_1^K$ and $\mathcal{D}_2(\ve{y}_2^N) = \hat{f}_2^K$.
\end{itemize}

We may suppress the $K,N$ subscripts for compact notation, writing $\mathfrak{C}$ instead, when there is no room for ambiguity.
The MSE for a computation code $\mathfrak{C}$ employed on the channel $\mathscr{H}$ is defined as
\begin{align}
    \mathrm{MSE}(\mathfrak{C}, \mathscr{H})  &\triangleq
    \sum_{j \in \{1,2\}, \ \kappa \in [1:K]} \mathbb{E} \left[ \lVert f(\kappa) - \hat{f}_j(\kappa)\rVert^2 \right],
\end{align}
where the expectation is over the data sequences and noise.
\end{definition}

\begin{definition}[AirComp DoF] We say that $(K,N)$ computation is feasible (in the DoF sense) on the channel $\mathscr{H}$ if for each $\rho$, there exists a computation code $\mathfrak{C}_{K,N}(\rho)$ such that $\mathrm{MSE}\!\left(\mathfrak{C}_{K,N}(\rho),\mathscr{H}\right)=O(1/\rho)$ as $\rho\rightarrow\infty$. A value $d\in \mathbb{R}_{\ge 0}$ is said to be an \emph{achievable AirComp DoF} on the channel $\mathscr{H}$ if there exist $K,N\in\mathbb{N}$, such that $d\leq K/N$, and $(K,N)$ computation is feasible on $\mathscr{H}$. The supremum of all achievable AirComp DoF values on the channel $\mathscr{H}$  is called the \emph{AirComp DoF} of  $\mathscr{H}$ and is denoted by $\mathrm{ACDoF}(\mathscr{H})$. For a random channel $\mathcal{H}$, we say that the AirComp DoF is \emph{almost surely} equal to $x$, and write   $\mathrm{ACDoF}(\mathcal{H}) \overset{\mathrm{a.s.}}{=} x,$
if $\Pr\!\left(\mathcal{H}\in \left\{\mathscr{H}:\ \mathrm{ACDoF}(\mathscr{H})=x\right\}\right)=1$.
\end{definition}

As a sanity check, Appendix~\ref{appendix_sisomac} shows that the AirComp DoF is precisely equal to $1$ for the conventional AirComp setting, i.e., a two-user SISO MAC obtained by setting $ \ve{y}_1= \ve{y}_2$ in \eqref{eq:modely1}.

\section{Results} \label{sec:results}
\subsection{Single-Input Single-Output (SISO) Setting}
We begin with the SISO setting, in which each transmitter and receiver is equipped with a single antenna.
\begin{theorem} \label{thm:1}
For a $(1,1;1,1;\mathscr{H})$ setting with $\m{H}(\nu)$ defined in \eqref{eq:a_H}, if
\begin{align}
h(\nu)\triangleq\frac{\m{H}_{11}(\nu)\m{H}_{22}(\nu)}
{\m{H}_{12}(\nu)\m{H}_{21}(\nu)}
\in\mathbb{C}\setminus\mathbb{R},\quad \forall\nu\in\mathbb{N},
\end{align}
then $\mathrm{ACDoF}(\mathscr{H}) = 2/3$.
\end{theorem}
In words, the SISO setting has AirComp DoF equal to $2/3$ if, at each channel use $\nu$, the ratio $h(\nu)$ is a complex number with a nonzero imaginary part. Note that the condition also implies that all channels are nonzero. The proof is provided in Section~\ref{sec:5}.

Theorem~\ref{thm:1} covers both time-invariant channels, for which $h_{ij}(\nu)=h_{ij}$ for all $\nu\in\mathbb{N}$, and arbitrarily time-varying channels. Although both cases have AirComp DoF $2/3$, their achievability schemes differ. Over $N=3$ channel uses, the time-invariant construction uses asymmetric complex signaling \cite{Cadambe_Jafar_Wang}: it is nonlinear over $\mathbb{C}$ but linear over $\mathbb{R}$, and requires the constant ratio $h(\nu)=h$ to have a nonzero imaginary part. If at least one of the three ratios $h(\nu)$ differs from the others, a scheme linear over $\mathbb{C}$ suffices. Theorem~\ref{thm:1} states the common sufficient condition $h(\nu)\in\mathbb{C}\setminus\mathbb{R}$ for both cases. Inspection of the time-varying construction and converse shows that this subcase only requires nonzero channel coefficients for achievability and $h(\nu)\neq1$ for the converse.

The following corollary is an immediate consequence of Theorem~\ref{thm:1}, because when the channels are drawn from a continuous distribution, the ratio $h(\nu)$ has a nonzero imaginary part almost surely.
\begin{corollary} \label{cor:siso_generic}
For a generic $(1,1;1,1;\mathcal{H})$ network, $\mathrm{ACDoF}(\mathcal{H}) \overset{\mathrm{a.s.}}{=} \tfrac{2}{3}$.
\end{corollary}

\subsection{Multiple-Input Multiple-Output (MIMO) Setting}
\noindent The AirComp DoF for a generic MIMO setting is characterized in the following theorem.
\begin{theorem} \label{thm:4}
For a generic $(M_1, M_2; N_1, N_2; \mathcal{H})$ network,
\begin{align}
	\mathrm{ACDoF}(\mathcal{H}) &\overset{\mathrm{a.s.}}{=} \min \Biggl\{M_1, M_2, N_1, N_2, \notag\\
	&\qquad\qquad\max\Bigl\{\tfrac{M_1+M_2}{3},\tfrac{N_1+N_2}{3}\Bigr\}\Biggr\}.
\end{align}
\end{theorem}
Achievability is proved in Section~\ref{proof:MIMOach}, and a tight information-theoretic converse is provided in Section~\ref{proof:MIMOconverse}.

\section{SISO Network: Proof of Theorem~\ref{thm:1}}\label{sec:5}
\subsection{Proof of Achievability}

Let us construct a coding scheme over $N=3$ channel uses. Starting with the model in \eqref{eq:modely1}, we will temporarily denote the original input, output and noise variables from Section II with a prime (e.g., $\bm{x}_1'$). To simplify our code design, we can represent the channel in the following normalized form,
\begin{align}
    \ve{y}_1(\nu) &= \ve{x}_1(\nu) + \ve{x}_2(\nu) + \ve{z}_1(\nu), \\
    \ve{y}_2(\nu) &= \ve{x}_1(\nu) + h (\nu) \ve{x}_2(\nu) + \ve{z}_2(\nu),
\end{align}
by defining the new un-primed effective variables as $\ve{x}_1(\nu)=\m{H}_{11}(\nu)\bm{x}_1'(\nu)$, $\ve{x}_2(\nu)=\m{H}_{12}(\nu)\bm{x}_2'(\nu)$, $\ve{y}_1(\nu)=\ve{y}_1'(\nu)$, $\ve{y}_2(\nu)=\tfrac{\m{H}_{11}(\nu)}{\m{H}_{21}(\nu)}\bm{y}_2'(\nu)$, $h(\nu) = \tfrac{\m{H}_{11}(\nu)\m{H}_{22}(\nu)}{\m{H}_{21}(\nu)\m{H}_{12}(\nu)}$, $\ve{z}_1(\nu)=\bm{z}_1'(\nu)$, and $\ve{z}_2(\nu)=\tfrac{\m{H}_{11}(\nu)}{\m{H}_{21}(\nu)}\bm{z}_2'(\nu)$.

 We distinguish between two cases. In the first case, the channel remains fixed (time-invariant) across the three channel uses, i.e., $h(1)=h(2)=h(3)$. In the second case, this condition does not hold, so that at least one of $h(1), h(2), h(3)$ differs from the other two; without loss of generality, let us assume $h(3) \notin \{h(1),h(2)\}$. We begin with the latter case. Throughout the analysis, we assume that $h(1),h(2),h(3)\in \mathbb{C}\setminus\mathbb{R}$.

{\bf Case:} $h(3)\notin\{h(1), h(2)\}$. The following table lists the inputs over the three  channel uses, along with the corresponding outputs at the two receivers. Noise terms are omitted for compact notation. The symbols $a_1,a_2,b_1,b_2$ will be defined shortly.
\begin{align}
\begin{array}{|c|c|c|c|}\hline
\nu& (\ve{x}_1(\nu),\ve{x}_2(\nu))& \ve{y}_1(\nu)&\ve{y}_2(\nu)\\\hline
\nu=1& (a_1,b_1)&a_1+b_1&a_1+h(1)b_1\\
\nu=2& (a_2,b_2)&a_2+b_2&a_2+h(2)b_2\\
\nu=3& (a_2,b_1)&a_2+b_1&a_2+h(3)b_1\\\hline
\end{array}
\end{align}
The key idea is that each receiver forms two carefully chosen linear combinations of its three observations so that both receivers recover the same two effective equations (up to noise), one for each computation instance. Each receiver combines its outputs as follows (noise terms are omitted). For the first computation instance,
\begin{align}
	\text{Rx-}1:\quad &\mbox{$\left(\frac{h(3)-1}{1-h(1)}\right)$}\ve{y}_1(1)+\ve{y}_1(3) \notag\\
	&\qquad=
	\mbox{$\left(\frac{h(3)-1}{1-h(1)}\right)$}a_1+a_2
	+\mbox{$\left(\frac{h(3)-h(1)}{1-h(1)}\right)$}b_1
	\label{eq:Rx1n1}\\
	\text{Rx-}2: \quad	&\mbox{$\left(\frac{h(3)-1}{1-h(1)}\right)$}\ve{y}_2(1)+\ve{y}_2(3) \notag\\
	&\qquad=
	\mbox{$\left(\frac{h(3)-1}{1-h(1)}\right)$}a_1+a_2
	+\mbox{$\left(\frac{h(3)-h(1)}{1-h(1)}\right)$}b_1
	\label{eq:Rx2n1}
\end{align}
and for the second computation instance,
\begin{align}
	\text{Rx-}1: \quad&\mbox{$\left(\frac{h(3)-1}{1-h(2)}\right)$}h(2)\ve{y}_1(2)+h(3)\ve{y}_1(3) \notag\\
	&\quad=
	\mbox{$\left(\frac{h(3)-h(2)}{1-h(2)}\right)$}a_2+h(3)b_1 +
	\mbox{$\left(\frac{h(3)-1}{1-h(2)}\right)$}h(2)b_2
	\label{eq:Rx1n2}\\
	\text{Rx-}2: \quad	&\mbox{$\left(\frac{h(3)-1}{1-h(2)}\right)$}\ve{y}_2(2)+\ve{y}_2(3) \notag\\
	&\quad=
	\mbox{$\left(\frac{h(3)-h(2)}{1-h(2)}\right)$}a_2+h(3)b_1 +
	\mbox{$\left(\frac{h(3)-1}{1-h(2)}\right)$}h(2)b_2
	\label{eq:Rx2n2}
\end{align}
Observe that the terms in the denominator are nonzero, e.g., $1-h(1)\neq 0$ because $h(1)\in\mathbb{C}\setminus\mathbb{R}$, whereas $1\notin \mathbb{C}\setminus\mathbb{R}$. Also note that the terms in the numerator are nonzero, e.g., $h(3)\neq h(1)$ because of the assumption that $h(3)\notin\{h(1),h(2)\}$. Comparing \eqref{eq:Rx1n1} with \eqref{eq:Rx2n1}, and \eqref{eq:Rx1n2} with \eqref{eq:Rx2n2}, observe that both receivers obtain the same linear combinations (plus noise terms that are not shown).
To ensure that the linear combinations represent the desired computation, we need the following,
\begin{align}
\mbox{$\left(\frac{h(3)-1}{1-h(1)}\right)$}a_1+a_2&=\mu s_1(1)\\
\mbox{$\left(\frac{h(3)-h(2)}{1-h(2)}\right)$}a_2&=\mu s_1(2)\\
\mbox{$\left(\frac{h(3)-h(1)}{1-h(1)}\right)$}b_1&=\mu s_2(1)\\
h(3)b_1+\mbox{$\left(\frac{h(3)-1}{1-h(2)}\right)$}h(2)b_2&=\mu s_2(2)
\end{align}
which is ensured by defining $a_1,a_2,b_1,b_2$ as follows.
\begin{align}
a_1&=\mu\mbox{$\left(\frac{1-h(1)}{h(3)-1}\right)$}\left(s_1(1)-\mbox{$\left(\frac{1-h(2)}{h(3)-h(2)}\right)$}s_1(2)\right)\\
a_2&=\mu\mbox{$\left(\frac{1-h(2)}{h(3)-h(2)}\right)$}s_1(2)\\
b_1&=\mu\mbox{$\left(\frac{1-h(1)}{h(3)-h(1)}\right)$}s_2(1)\\
b_2&=\mu\mbox{$\left(\frac{1-h(2)}{(h(3)-1)h(2)}\right)$}\left(s_2(2)-h(3)\mbox{$\left(\frac{1-h(1)}{h(3)-h(1)}\right)$}s_2(1)\right).
\end{align}
Here, $\mu>0$ is a scaling factor  chosen to ensure that the transmit power constraint is not violated. It depends on the values $h(\nu)$, but in particular does not depend on the noise  parameter $\sigma^2$. By the same token, the noise present in each computation output is scaled by a factor that depends on channel coefficients but not on $\sigma^2$. Thus the MSE achieved by the coding scheme is $O(\sigma^2)$. Since $K=2$ instances of the desired computation are achieved over $N=3$ channel uses, the AirComp DoF value achieved  is $2/3$. We now proceed to the remaining case.

{\bf Case:} $h(1)=h(2)=h(3)=\alpha+j\beta$, where $\alpha,\beta\in\mathbb{R}$ and $\beta\neq 0$. Our solution for this case is inspired by the idea of asymmetric complex signaling that was previously introduced for interference alignment schemes in \cite{Cadambe_Jafar_Wang}. Our design treats each complex symbol as a two-dimensional real vector and constructs the code linearly over the real domain. Thus, while the scheme is $\mathbb{R}$-linear, it is generally not $\mathbb{C}$-linear after converting back to the complex representation. Let $\mathcal{T}: \mathbb{C} \to \mathbb{R}^{2 \times 1}$ and $\mathcal{W}: \mathbb{C} \to \mathbb{R}^{2 \times 2}$ be real-valued operators defined as:
\begin{align}
    \mathcal{T}(X) &\triangleq \begin{bsmallmatrix} \Re\{X\} \\ \Im\{X\}  \end{bsmallmatrix}, &&
    \mathcal{W}(h) \triangleq \begin{bsmallmatrix} \Re\{h\} & -\Im\{h\} \\ \Im\{h\} & \Re\{h\} \end{bsmallmatrix}.
\end{align}
Thus, we can write the network equations as
\begin{align}
    \mathcal{T}(\ve{y}_1(\nu)) &= \mathcal{T}(\ve{x}_1(\nu)) + \mathcal{T}(\ve{x}_2(\nu)) + \mathcal{T}(\ve{z}_1(\nu)), \\
    \mathcal{T}(\ve{y}_2(\nu)) &= \mathcal{T}(\ve{x}_1(\nu)) + \mathcal{W}(h) \mathcal{T}(\ve{x}_2(\nu)) + \mathcal{T}(\ve{z}_2(\nu)).
\end{align}
Stacking the equations over $3$ channel uses yields
\begin{align}
    \mathcal{T}(\ve{y}_1)^{(3)}
    &=
    \mathcal{T}(\ve{x}_1)^{(3)}
    +
    \mathcal{T}(\ve{x}_2)^{(3)}
    +
    \mathcal{T}(\ve{z}_1)^{(3)},\\
    \mathcal{T}(\ve{y}_2)^{(3)}
    &=
    \mathcal{T}(\ve{x}_1)^{(3)}
    +
    \mathcal{W}(h)^{[3]}
    \mathcal{T}(\ve{x}_2)^{(3)}
    +
    \mathcal{T}(\ve{z}_2)^{(3)}.
\end{align}
The encoders are defined as  $\mathcal{T}(\ve{x}_1)^{(3)} = \m{E}_1 \mathcal{T}(s_1)^{(2)}$, and $\mathcal{T}(\ve{x}_2)^{(3)} = \m{E}_2 \mathcal{T}(s_2)^{(2)}$, where
\begin{align}
    \m{E}_1 &= \mu
    \begin{bmatrix}
        \ve{e}_2 + \ve{e}_6 & \ve{e}_4 + \ve{e}_5 & \alpha \ve{e}_1 +  \beta\ve{e}_2 & \alpha \ve{e}_3 + \beta \ve{e}_4
    \end{bmatrix}, \\
    \m{E}_2 &= \mu
    \begin{bmatrix}
        \ve{e}_2 + \ve{e}_6 & \ve{e}_4 + \ve{e}_5 & \ve{e}_1 & \ve{e}_3
    \end{bmatrix}.
\end{align}
Here $\ve{e}_i$ is the standard basis vector corresponding to the $i$-th column of the $6\times 6$ identity matrix.
To successfully recover the desired sum over the real domain, the encoders $\m{E}_1, \m{E}_2 \in \mathbb{R}^{6 \times 4}$ and decoders $\m{D}_1, \m{D}_2 \in \mathbb{R}^{6 \times 4}$ must satisfy
\begin{align}
    \m{D}_1^T \m{E}_1 &= \m{D}_1^T \m{E}_2 = \m{I}_4, \label{eq:align_rx1}\\
    \m{D}_2^T \m{E}_1 &= \m{D}_2^T \mathcal{W}(h)^{[3]} \m{E}_2 = \m{I}_4. \label{eq:align_rx2}
\end{align}
These identities yield the necessary null-space conditions:
\begin{align}
    \m{D}_1^T (\m{E}_1 - \m{E}_2) = \m{0}, \quad \m{D}_2^T \left(\m{E}_1 - \mathcal{W}(h)^{[3]} \m{E}_2\right) = \m{0}.
\end{align}
Because the total space is $6$-dimensional and the decoders project onto a $4$-dimensional subspace, the null space of each decoder has a dimension of $6 - 4 = 2$. The matrices $\m{E}_1$ and $\m{E}_2$ are designed so that the column span of $\m{E}_1 - \m{E}_2$ lies within the null space of $\m{D}_1$, and the column span of $\m{E}_1 - \mathcal{W}(h)^{[3]} \m{E}_2$ lies within the null space of $\m{D}_2$. Since these null spaces are 2-dimensional, the difference matrices $\m{E}_1 - \m{E}_2$ and $\m{E}_1 - \mathcal{W}(h)^{[3]} \m{E}_2$ must each have a column rank of at most 2.

To satisfy these rank constraints, the matrix columns are aligned directly: the first two columns of $\m{E}_1$ are identical to the first two columns of $\m{E}_2$; the last two columns of $\m{E}_1$ are equal to the matrix $\mathcal{W}(h)^{[3]}$ multiplied by the corresponding columns of $\m{E}_2$.

With the encoders fixed, the decoders are constructed in two steps. First, choose
$\tilde{\m{D}}_1^T \in \mathbb{R}^{4\times 6}$ and
$\tilde{\m{D}}_2^T \in \mathbb{R}^{4\times 6}$ such that
\begin{align}
    \tilde{\m{D}}_1^T(\m{E}_1-\m{E}_2) &= \m{0},\\
    \tilde{\m{D}}_2^T\bigl(\m{E}_1-\mathcal{W}(h)^{[3]}\m{E}_2\bigr) &= \m{0}.
\end{align}
That is, the rows of $\tilde{\m{D}}_1^T$ and $\tilde{\m{D}}_2^T$ are chosen from the
left null spaces of the corresponding difference matrices. By construction,
\begin{align}
    \tilde{\m{D}}_1^T \m{E}_1 &= \tilde{\m{D}}_1^T \m{E}_2,\\
    \tilde{\m{D}}_2^T \m{E}_1 &= \tilde{\m{D}}_2^T \mathcal{W}(h)^{[3]} \m{E}_2.
\end{align}
Provided these effective $4\times 4$ matrices are invertible, we normalize them by defining
\begin{align}
    \m{D}_1^T &= \bigl(\tilde{\m{D}}_1^T\m{E}_1\bigr)^{-1}\tilde{\m{D}}_1^T,\\
    \m{D}_2^T &= \bigl(\tilde{\m{D}}_2^T\m{E}_1\bigr)^{-1}\tilde{\m{D}}_2^T.
\end{align}
It then follows immediately that
\begin{align}
    \m{D}_1^T \m{E}_1 &= \m{D}_1^T \m{E}_2 = \m{I}_4,\\
    \m{D}_2^T \m{E}_1 &= \m{D}_2^T \mathcal{W}(h)^{[3]} \m{E}_2 = \m{I}_4.
\end{align}
The decoding matrices are given by:
\begin{align}
	\m{D}_1^T &= \frac{1}{\mu} \tilde{\m{D}}_1^T \triangleq \frac{1}{\mu}
    \begin{bsmallmatrix}
        -(\alpha-1) & 0 & \beta & 0 \\
        0 & -(\alpha-1) & 0 & \beta \\
        0 & 1 & 0 & 0 \\
        1 & 0 & 0 & 0
    \end{bsmallmatrix}^{-1}
    \begin{bsmallmatrix}
       \beta \ve{e}_1^T - (\alpha-1) \ve{e}_2^T \\
       \beta \ve{e}_3^T - (\alpha-1) \ve{e}_4^T \\
        \ve{e}_5^T \\
        \ve{e}_6^T
    \end{bsmallmatrix},\\
    \m{D}_2^T &= \frac{1}{\mu} \tilde{\m{D}}_2^T \triangleq \frac{1}{\mu}
    \begin{bsmallmatrix}
        -\beta & 0 & \Delta & 0 \\
        0 & -\beta & 0 & \Delta \\
        0 & 0 & -\beta & \alpha \\
        0 & 0 & \alpha & \beta
    \end{bsmallmatrix}^{-1}
    \begin{bsmallmatrix}
        (1-\alpha)\ve{e}_1^T -\beta \ve{e}_2^T \\
        (1-\alpha)\ve{e}_3^T - \beta \ve{e}_4^T \\
        -\ve{e}_2^T + \ve{e}_3^T + \ve{e}_6^T\\
        \ve{e}_1^T + \ve{e}_4^T - \ve{e}_5^T
    \end{bsmallmatrix},
\end{align}
where $\Delta \triangleq (1-\alpha)\alpha-\beta^2$. The existence of the inverses in $\m{D}_1$ and $\m{D}_2$ is guaranteed by the assumption that the channel ratio is strictly complex, ensuring $\beta \neq 0$. The scaling parameter $\mu$ is chosen so that the transmit power constraint is satisfied. Note that $\mu$ depends only on the channels and not on $\sigma^2$. Since the decoding operations scale the noise by channel-dependent factors, the MSE is bounded by $c\sigma^2=c/\rho$ for some channel-dependent constant $c$ independent of $\sigma^2$. Thus $\mathrm{MSE}(\mathfrak{C}, \mathscr{H})=O(1/\rho)$. Since $K=2$ desired computations are achieved over $N=3$ channel uses, the AirComp DoF achieved by this coding scheme is $2/3$.

\subsection{Proof of Converse}

Consider a fixed, deterministic channel realization $\mathscr{H}$ satisfying the hypothesis of Theorem~\ref{thm:1}, namely that for all channel uses $\nu$, the ratio $h(\nu) \in \mathbb{C} \setminus \mathbb{R}$. Suppose that for each SNR $\rho$, there exists a computation code $\mathfrak{C}_{\rho}$ of length $N_\rho$ encoding $K_\rho$ symbols such that $\mathrm{MSE}(\mathfrak{C}_{\rho},\mathscr{H}) = O(1/\rho)$. Let $d=\lim_{\rho\to\infty}K_\rho/N_\rho$. The code and the specific channel realization are assumed to be known globally by all transmitters and receivers.

Following the MSE--equivocation bounding technique established in steps \eqref{eq:mac_1}--\eqref{eq:mac_4} of the proof of Lemma~\ref{lemma:SISOMAC} in Appendix~\ref{appendix_sisomac}, the mutual information between the desired function and the received signal at each receiver is bounded below by:
\begin{align}
    I(f^{K_\rho};\ve{y}_i^{N_\rho})
    &\ge K_\rho \log\rho + K_\rho o(\log \rho),
    \quad i \in \{1,2\}. \label{eq:conv_lower_bound}
\end{align}
Summing the bounds for both receivers and expanding the mutual information terms, we obtain:
\begin{align}
    2K_\rho \log\rho &+ K_\rho o(\log \rho) \notag\\
    &\le I(f^{K_\rho};\ve{y}_1^{N_\rho})
    + I(f^{K_\rho};\ve{y}_2^{N_\rho}) \label{eq:conv_sum_I} \\
    &= h(\ve{y}_1^{N_\rho}) + h(\ve{y}_2^{N_\rho}) \notag\\
    &\quad - h(\ve{y}_1^{N_\rho}| f^{K_\rho})
    - h(\ve{y}_2^{N_\rho}| f^{K_\rho}). \label{eq:conv_sum_h}
\end{align}
To bound the unconditional entropy terms in \eqref{eq:conv_sum_h}, we apply the power bound derived in Appendix~\ref{appendix_4}:
\begin{align} \label{eq:conv_power_bound}
    h(\ve{y}_1^{N_\rho}) + h(\ve{y}_2^{N_\rho})
    &\le 2N_\rho \log \left(2M_o^2 + \frac{1}{\rho} \right) \notag\\
    &\quad + N_\rho o(\log \rho).
\end{align}
For the remaining conditional entropy terms in \eqref{eq:conv_sum_h}, we apply the chain rule and condition on $s_1^{K_\rho}$:
\begin{align}
    &- h(\ve{y}_1^{N_\rho}| f^{K_\rho})
    - h(\ve{y}_2^{N_\rho}| f^{K_\rho})
    \le -h(\ve{y}_1^{N_\rho},\ve{y}_2^{N_\rho}| f^{K_\rho})  \\
    &= - I(\ve{y}_1^{N_\rho}, \ve{y}_2^{N_\rho};
    s_1^{K_\rho}| f^{K_\rho})
    - h(\ve{y}_1^{N_\rho},\ve{y}_2^{N_\rho}|
    s_1^{K_\rho}, f^{K_\rho}) \\
    &= h(s_1^{K_\rho}|\ve{y}_1^{N_\rho},\ve{y}_2^{N_\rho},
    f^{K_\rho}) - h(s_1^{K_\rho}|f^{K_\rho})
    - h(\ve{z}_1^{N_\rho}, \ve{z}_2^{N_\rho}).
    \label{eq:conv_cond_expansion}
\end{align}
Note that the joint noise entropy is $h(\ve{z}_1^{N_\rho}, \ve{z}_2^{N_\rho}) = -2N_\rho \log \rho + N_\rho o(\log \rho)$, and $h(s_1^{K_\rho}|f^{K_\rho}) = K_\rho h(s_1|f) = K_\rho o(\log \rho)$. Substituting these into \eqref{eq:conv_cond_expansion} yields
\begin{align}
    &- h(\ve{y}_1^{N_\rho}| f^{K_\rho})
    - h(\ve{y}_2^{N_\rho}| f^{K_\rho}) \notag\\
    &\le h(s_1^{K_\rho}|\ve{y}_1^{N_\rho},\ve{y}_2^{N_\rho},
    f^{K_\rho}) + 2 N_\rho \log \rho \notag\\
    &\quad +K_\rho o(\log\rho)+N_\rho o(\log\rho).
    \label{eq:rembound}
\end{align}

Next, we bound the term $h(s_1^{K_\rho}|\ve{y}_1^{N_\rho},\ve{y}_2^{N_\rho}, f^{K_\rho})$. Because the channel ratio satisfies $h(\nu) \in\mathbb{C}\setminus\mathbb{R}$, we strictly have $\m{H}_{11}(\nu)\m{H}_{22}(\nu) \neq \m{H}_{12}(\nu)\m{H}_{21}(\nu)$. This guarantees that the block channel matrix $\m{H}(\nu)$ has full rank and is invertible. Applying the Invertibility Bound (Lemma~\ref{lm:invert} in Appendix~\ref{app:invert}), we get
\begin{align}
    h(s_1^{K_\rho}|\ve{y}_1^{N_\rho},\ve{y}_2^{N_\rho},
    f^{K_\rho})
    &\le -K_\rho \log \rho + K_\rho o(\log \rho) \notag\\
    &\quad +N_\rho o(\log \rho).
    \label{eq:invert_bound_app}
\end{align}
Substituting \eqref{eq:invert_bound_app} into \eqref{eq:rembound} gives:
\begin{align} \label{eq:conv_4.5.1}
    - h(\ve{y}_1^{N_\rho}|f^{K_\rho})
    - h(\ve{y}_2^{N_\rho}|f^{K_\rho})
    &\le (2N_\rho - K_\rho) \log \rho \notag\\
    &\quad +K_\rho o(\log\rho)+N_\rho o(\log\rho).
\end{align}
Finally, substituting the bounds from \eqref{eq:conv_power_bound} and \eqref{eq:conv_4.5.1} back into \eqref{eq:conv_sum_I}--\eqref{eq:conv_sum_h}, we obtain,
\begin{align}
    &2K_\rho \log\rho + K_\rho o(\log \rho) \notag\\
    &\le 2N_\rho \log \left(2M_o^2 + \frac{1}{\rho} \right)
    + (2N_\rho - K_\rho) \log \rho \notag\\
    &\quad +K_\rho o(\log\rho)+N_\rho o(\log\rho).
    \label{eq:conv_final_ineq}
\end{align}
Dividing both sides by $N_\rho\log\rho$ and taking $\rho \to \infty$, the lower-order terms vanish because $K_\rho/N_\rho$ has a finite limit. Hence, $d\le2/3$, which completes the proof of the converse.

\section{MIMO Generic Network}
When an $(M_1,M_2;N_1,N_2;\mathscr{H})$ MIMO Gaussian network is operated over channel uses $\nu=1,\ldots,N$, the corresponding input–output relations can be written as,
\begin{align}
    \ve{y}_i^{(N)} &= \m{H}_{i1}^{[N]}\ve{x}_1^{(N)} + \m{H}_{i2}^{[N]}\ve{x}_2^{(N)} + \ve{z}_i^{(N)},
\end{align}
where  $\ve{y}_i^{(N)}\in\mathbb{C}^{N N_i}, \ve{x}_i^{(N)}\in\mathbb{C}^{N M_i}$ for $i\in\{1,2\}$.
We will follow this notation in this proof.

\subsection{Proof of Achievability}\label{proof:MIMOach}
Let us begin with \new{virtual} subnetwork packing. \new{This decomposition is not a physical partition of a generic channel. Instead, it constructs one realization of the channel matrices that decomposes into the packed subnetworks and makes the relevant rational polynomials nonzero. These polynomials are therefore nonzero almost surely for a generic channel (Lemma~\ref{lem:achievability_packing}).}
\begin{definition}[\new{Virtual} Subnetwork Packing] \label{def:filled}
We say that an $(M_1,M_2;N_1,N_2)$ network can be \emph{filled} with $k(t)$ copies of a subnetwork of type $(m_1(t),m_2(t);n_1(t),n_2(t))$ for $t\in[1:T]$ if the antenna allocations satisfy the component-wise inequality:
\begin{align}
\sum_{t=1}^{T} & k(t)
\begin{bmatrix}
m_1(t) & m_2(t) & n_1(t) & n_2(t)
\end{bmatrix} \notag\\
&\le
\begin{bmatrix}
M_1 & M_2 & N_1 & N_2
\end{bmatrix}.\label{eq:packingconstraint}
\end{align}
\end{definition}
The following lemma establishes an AirComp DoF achievability result for any valid subnetwork composition.

\begin{lemma}[Achievability via Packing] \label{lem:achievability_packing}
Suppose a generic network can be filled with $a$ copies of a $(1,1;1,1)$ subnetwork, $b$ copies of a $(1,1;2,1)$ subnetwork, and $c$ copies of a $(1,1;1,2)$ subnetwork. Then, an overall AirComp DoF of $\frac{2}{3}a + b + c$ is achievable almost surely.
\end{lemma}
The proof of Lemma~\ref{lem:achievability_packing} is provided in Appendix~\ref{appendix_sisomacchievability}. \new{Table~\ref{tab:packing_notation} in Appendix~\ref{appendix_sisomacchievability} summarizes the packing parameters and block matrices used in the construction and its proof.}

 Having established that the quantity $\frac{2}{3}a + b + c$ represents an achievable AirComp DoF, our next step is to prove that we can always find a nonnegative integer allocation $(a,b,c)$ that satisfies the packing constraint \eqref{eq:packingconstraint} while reaching the maximum AirComp DoF of the network. We begin with the regime $M_1 + M_2 \le N_1 + N_2$, i.e., the \emph{receive-antenna-abundant} regime.

\begin{lemma}[Subnetwork Allocation] \label{lem:packing_rx_heavy}
For an $(M_1,M_2;N_1,N_2)$ network satisfying $M_1 + M_2 \leq N_1 + N_2$, there exist nonnegative integers $a, b, c$ such that the network can be filled with $a$ copies of a $(1,1;1,1)$ subnetwork, $b$ copies of a $(1,1;2,1)$ subnetwork, and $c$ copies of a $(1,1;1,2)$ subnetwork, and
	$\frac{2}{3}a + b + c = \min\left\{M_1, M_2, N_1, N_2, \frac{N_1 + N_2}{3}\right\}$.
\end{lemma}

\begin{proof}
We explicitly provide the integer allocations $(a, b, c)$ that satisfy Definition~\ref{def:filled} while achieving the target sum, in each of the following cases.\\
\noindent \textbf{Case 1:} The total number of receive antennas acts as the bottleneck, i.e., $\min\left\{M_1, M_2, N_1, N_2, \frac{N_1 + N_2}{3}\right\} = \frac{N_1 + N_2}{3}$. We allocate the parameters as $a = 3\left\lceil \frac{N_1+N_2}{3}\right\rceil - (N_1+N_2)$, \new{$b = N_1 - \left\lceil \frac{N_1+N_2}{3}\right\rceil$, $c= N_2 - \left\lceil \frac{N_1+N_2}{3}\right\rceil$}.

\noindent \textbf{Case 2:} One specific receiver acts as the bottleneck, e.g., $\min\left\{M_1, M_2, N_1, N_2, \frac{N_1 + N_2}{3}\right\} = N_1$ (the case for $N_2$ follows by symmetry). We set \new{$a=0, b=0, c=N_1$}.

\noindent \textbf{Case 3:} One specific transmitter acts as the bottleneck, e.g., $\min\left\{M_1, M_2, N_1, N_2, \frac{N_1 + N_2}{3}\right\} = M_1$ (the case for $M_2$ follows by symmetry). We set $a=0,b = \max \{0, 2M_1 - N_2\}, c= M_1 - b$.
In all three cases, direct substitution verifies that the \new{virtual} subnetwork packing constraints are satisfied, and the desired metric $\frac{2}{3}a + b + c$ is attained.
\end{proof}

The preceding lemmas jointly establish the achievability for the receive-antenna-abundant regime. We now turn to the complementary transmit-antenna-abundant regime: $M_1+M_2 > N_1+N_2$.

By exploiting the structural symmetry of the network (effectively reversing the roles of transmitters and receivers), we can apply a dual packing strategy. In this regime, the network is filled with $a$ copies of a $(1,1;1,1)$ subnetwork, $b$ copies of a $(2,1;1,1)$ subnetwork, and $c$ copies of a $(1,2;1,1)$ subnetwork. Following a constructive logic parallel to Lemmas~\ref{lem:achievability_packing} and~\ref{lem:packing_rx_heavy}, this dual packing achieves an AirComp DoF of
    $\frac{2}{3}a + b + c = \min\left\{M_1, M_2, N_1, N_2, \frac{M_1 + M_2}{3}\right\}$.
Combining the two regimes completes the proof of achievability for Theorem~\ref{thm:4}.

\subsection{Proof of Converse}\label{proof:MIMOconverse}
We first derive the cooperative upper bound $\mathrm{ACDoF}(\mathcal{H}) \overset{\mathrm{a.s.}}{\le} \frac{1}{3}\max\{M_1 + M_2, N_1 + N_2\}$. \new{For intuition, suppose $N_1+N_2\ge M_1+M_2$. Delivering $K$ sums to each receiver accounts for $2K$ desired computation dimensions, while network invertibility forces the joint observations to retain one additional independent $K$-dimensional source combination. Hence, $3K$ signal dimensions must fit within the $N(N_1+N_2)$-dimensional joint receive space. The proof below formalizes this signal-space count using MSE--equivocation bounds and invertibility bounds; it does not require either receiver to decode the individual sources.} Consider a generic channel realization $\mathscr{H}$ satisfying $N_1 + N_2 \ge M_1 + M_2$. \new{Suppose that for each SNR $\rho$, there exists a computation code $\mathfrak{C}_{\rho}$ of length $N_\rho$ encoding $K_\rho$ symbols such that $\mathrm{MSE}(\mathfrak{C}_{\rho},\mathscr{H}) = O(1/\rho)$. Let $d=\lim_{\rho\to\infty}K_\rho/N_\rho$.}
For a generic channel, the block channel matrix $\m{H}^{[N_\rho]}$ has full column rank almost surely.
 The code and channel are assumed to be known globally.

Applying the MSE--equivocation bound (\new{Lemma~\ref{lm:mse_equiv} in} Appendix~\ref{appendix_mse}), the mutual information at Rx-$1$ is bounded by:
\begin{align} \label{eq:mimo_conv_I1}
    I(f^{K_\rho};\ve{y}_1^{N_\rho})
    &= h(f^{K_\rho}) - h(f^{K_\rho}|\ve{y}_1^{N_\rho}) \notag\\
    &\ge h(f^{K_\rho}) + K_\rho \log\rho
    +K_\rho o(\log \rho) \notag\\
    &= K_\rho \log\rho + K_\rho o(\log \rho).
\end{align}
Similarly, for Rx-$2$ we have:
\begin{align} \label{eq:mimo_conv_I2}
    I(f^{K_\rho};\ve{y}_2^{N_\rho})
    \ge K_\rho \log\rho + K_\rho o(\log \rho).
\end{align}
Summing \eqref{eq:mimo_conv_I1} and \eqref{eq:mimo_conv_I2} yields:
\begin{align}
    2K_\rho \log\rho &+ K_\rho o(\log \rho) \notag \\
    &\le I(f^{K_\rho};\ve{y}_1^{N_\rho})
    + I(f^{K_\rho};\ve{y}_2^{N_\rho}) \\
    &= h(\ve{y}_1^{N_\rho}) + h(\ve{y}_2^{N_\rho}) \notag\\
    &\quad - h(\ve{y}_1^{N_\rho}|f^{K_\rho})
    - h(\ve{y}_2^{N_\rho}|f^{K_\rho}).
    \label{eq:mimo_conv_sum}
\end{align}
From the power bound derived in Appendix~\ref{appendix_4},
\begin{align} \label{eq:mimo_conv_power}
    &h(\ve{y}_1^{N_\rho}) + h(\ve{y}_2^{N_\rho}) \notag \\
    &\le N_\rho (N_1 + N_2)
    \log \left((M_1+M_2)M_o^2 + 1/\rho\right) \notag\\
    &\quad +N_\rho o(\log \rho).
\end{align}
To bound the remaining terms in \eqref{eq:mimo_conv_sum}, we apply the chain rule and condition on the source sequence $s_1^{K_\rho}$:
\begin{align}
    &- h(\ve{y}_1^{N_\rho}|f^{K_\rho})
    - h(\ve{y}_2^{N_\rho}|f^{K_\rho}) \notag\\
    &\le -h(\ve{y}_1^{N_\rho},\ve{y}_2^{N_\rho}|f^{K_\rho})  \\
    &= - I(\ve{y}_1^{N_\rho}, \ve{y}_2^{N_\rho};
    s_1^{K_\rho}| f^{K_\rho}) \notag\\
    &\quad - h(\ve{y}_1^{N_\rho},\ve{y}_2^{N_\rho}|
    s_1^{K_\rho}, f^{K_\rho})  \\
    &= h(s_1^{K_\rho}|\ve{y}_1^{N_\rho},\ve{y}_2^{N_\rho},
    f^{K_\rho}) - h(s_1^{K_\rho}|f^{K_\rho}) \notag\\
    &\quad - h(\ve{z}_1^{N_\rho}, \ve{z}_2^{N_\rho})  \\
    &= h(s_1^{K_\rho}|\ve{y}_1^{N_\rho},\ve{y}_2^{N_\rho},
    f^{K_\rho}) - h(s_1^{K_\rho}|f^{K_\rho}) \notag \\
    &\quad + N_\rho(N_1 + N_2) \log \rho
    +N_\rho o(\log \rho). \label{eq:mimo_conv_cond}
\end{align}
As established previously, $h(s_1^{K_\rho}|f^{K_\rho})=K_\rho h(s_1|f)=K_\rho o(\log\rho)$. Furthermore, because $N_1 + N_2 \ge M_1 + M_2$, the stacked channel matrix $\m{H}^{[N_\rho]}$ has full column rank almost surely. This allows us to apply the invertibility bound (Lemma~\ref{lm:invert}):
\begin{align} \label{eq:mimo_conv_invert}
    h(s_1^{K_\rho}|\ve{y}_1^{N_\rho},\ve{y}_2^{N_\rho},
    f^{K_\rho})
    &\le -K_\rho \log \rho+K_\rho o(\log \rho) \notag\\
    &\quad +N_\rho o(\log \rho).
\end{align}
Substituting \eqref{eq:mimo_conv_invert} into \eqref{eq:mimo_conv_cond} gives:
\begin{align} \label{eq:mimo_conv_cond_final}
    - h(\ve{y}_1^{N_\rho}|f^{K_\rho})
    &- h(\ve{y}_2^{N_\rho}|f^{K_\rho}) \notag \\
    &\le [N_\rho(N_1 + N_2) - K_\rho] \log \rho \notag\\
    &\quad +K_\rho o(\log\rho)+N_\rho o(\log\rho).
\end{align}
Finally, substituting the bounds from \eqref{eq:mimo_conv_power} and \eqref{eq:mimo_conv_cond_final} back into \eqref{eq:mimo_conv_sum}, we obtain the unified inequality:
\begin{align}
    2K_\rho \log\rho &+ K_\rho o(\log \rho) \notag \\
    &\le N_\rho (N_1 + N_2)
    \log \left((M_1+M_2)M_o^2 + 1/\rho\right) \notag \\
    &\quad + [N_\rho(N_1 + N_2) - K_\rho] \log \rho \notag\\
    &\quad +K_\rho o(\log\rho)+N_\rho o(\log\rho).
    \label{eq:mimo_conv_final_ineq}
\end{align}
Dividing both sides by $N_\rho \log\rho$ and taking $\rho \to \infty$, the lower-order terms vanish because $K_\rho/N_\rho$ has a finite limit. We obtain
    $d \le \frac{N_1 + N_2}{3}$,
which proves that when $N_1 + N_2 \ge M_1 + M_2$,
\begin{align} \label{eq:ub1}
    \mathrm{ACDoF}(\mathcal{H}) \overset{\mathrm{a.s.}}{\le} \frac{1}{3} ( N_1 + N_2 ).
\end{align}

To address the complementary transmit-antenna-abundant regime where $N_1+N_2 < M_1+M_2$, we employ an antenna-augmentation argument. In this regime, the generic channel matrix $\m{H}(\nu) \in \mathbb{C}^{(N_1+N_2) \times (M_1+M_2)}$ almost surely has full row rank. We can augment the network by introducing $M_1+M_2-(N_1+N_2)$ virtual receive antennas (each with independent Gaussian noise), appending linearly independent rows to $\m{H}(\nu)$ to form a square, full-rank matrix $\m{H}^\prime(\nu) \in \mathbb{C}^{(M_1+M_2) \times (M_1+M_2)}$. \new{Since providing additional observations to the receivers cannot decrease the achievable AirComp DoF, the AirComp DoF of the original network is upper-bounded by that of the augmented network.} Because the augmented network satisfies the full-column-rank condition, we can invoke the result from \eqref{eq:ub1} to conclude:
\begin{align} \label{eq:ub2}
    \mathrm{ACDoF}(\mathcal{H}) \overset{\mathrm{a.s.}}{\le} \frac{1}{3} (M_1 + M_2),
\end{align}
when $N_1 + N_2 < M_1 + M_2$.

\new{Because the receivers need not decode the individual messages, we first use a genie to reduce computation to point-to-point communication. Since $f^K$ and $s_2^K$ together determine $s_1^K$, revealing $s_2^K$ to a receiver turns any computation code for $f^K$ into a communication code for $s_1^K$ over the corresponding individual link. The single-receiver MIMO-MAC converse of Remark~\ref{rem:mimo_mac} then bounds this genie-aided problem. Repeating the argument for both sources and both receivers yields:}
\begin{align} \label{eq:c2}
    \mathrm{ACDoF}(\mathcal{H}) \overset{\mathrm{a.s.}}{\le} \min\{M_1, M_2, N_1, N_2\}.
\end{align}
The cooperative upper bounds \eqref{eq:ub1} and \eqref{eq:ub2} together with the individual link bound \eqref{eq:c2}
complete the proof of the converse for Theorem~\ref{thm:4}.

\section{MSE at Finite SNR: Alternating Optimization}\label{sec:numerical}

While the preceding sections establish the asymptotic AirComp DoF behavior (i.e., $\mathrm{MSE} = O(1/\rho)$ as $\rho \to \infty$) for the Gaussian network, we now use finite-SNR simulations to examine how the theoretical thresholds appear in the MSE behavior. Let us focus on a deterministic \new{$(M_1,M_2;N_1,N_2; \mathscr{H})$} MIMO network. As a representative finite-SNR transceiver, we consider an iterative alternating optimization (AO) algorithm inspired by transceiver designs for interference channels \cite{gomadamDistributedNumericalApproach2011} and federated learning \cite{huRISAssistedOvertheAirFederated2022}. \new{We benchmark it against three baselines: AirComp-time-sharing, BC-time-sharing, and a compound multiple-access channel (C-MAC) source-recovery scheme.}

\new{Let $\mathrm{ACDoF}^{\star}$ denote the AirComp DoF of a generic channel with the given antenna tuple, as characterized by Theorem~\ref{thm:4}; the superscript $\star$ emphasizes its optimality over all schemes. This value is attained by the construction of Section~\ref{proof:MIMOach}. The AO method below is a separate finite-SNR linear design. We neither prove that AO attains this asymptotic benchmark nor present it as an AirComp-DoF-achieving transceiver for all channel realizations. The numerical study has a narrower purpose: to examine whether its finite-SNR MSE behavior is consistent with the known theoretical thresholds. The observed agreement is empirical and is not an achievability argument. Developing a practical transceiver whose high-SNR behavior is proved to attain the AirComp DoF across channel realizations is left for future work.}

Let $L$ denote the number of simultaneous AirComp tasks per channel use. We restrict our attention to linear encoding matrices $\m{E}_i \in \mathbb{C}^{M_i \times L}$ and linear decoding matrices $\m{D}_j \in \mathbb{C}^{N_j \times L}$ for $i,j \in \{1,2\}$. We assume the data symbols $\ve{s}_i \in \mathbb{C}^L$ are drawn independently with zero mean and normalized covariance $\E{\ve{s}_i \ve{s}_i^\dagger} = \sigma_s^2 \m{I}_L$. \new{Tx-$i$ transmits $\ve{x}_i=\m{E}_i\ve{s}_i$, and Rx-$j$ applies its decoder to the received signal to form the estimate}
\begin{align}\label{eq:fhat_linear}
\new{\hat{\ve{f}}_j = \m{D}_j^\dagger\ve{y}_j = \m{D}_j^\dagger\Big(\sum_{i=1}^{2}\m{H}_{ji}\m{E}_i\ve{s}_i+\ve{z}_j\Big),\qquad j\in\{1,2\}.}
\end{align}
Without loss of generality, we normalize the \new{symbol variance to $\sigma_s^2 = 1$} and the transmit power budget to $1$, such that the power constraint at each transmitter is $\E{\|\ve{x}_i\|^2} = \|\m{E}_i\|_F^2 \le 1$.

\subsection{MSE Objective}

Recall the formal MSE metric defined in Section~\ref{sec:problem_formulation}. Under this linear precoding and decoding framework, the network MSE can be explicitly parameterized and expanded as a function of matrices $\{\m{E}_i\}$ and $\{\m{D}_j\}$:
\begin{align} \label{eq:MSE_expanded}
\mathrm{MSE}&(\{\m{E}_i\}, \{\m{D}_j\}) \notag \\
&= \sum_{j=1}^2 \mathbb{E} \left\| \hat{\ve{f}}_j - (\ve{s}_1 + \ve{s}_2) \right\|^2  \\
&= \sum_{j=1}^2 \mathbb{E} \left\| (\m{D}_j^\dagger \m{H}_{j1}\m{E}_1 - \m{I}_L)\ve{s}_1 \right.  \notag\\
&\quad \left. + (\m{D}_j^\dagger \m{H}_{j2}\m{E}_2 - \m{I}_L)\ve{s}_2 + \m{D}_j^\dagger \ve{z}_j \right\|^2  \\
&= \sigma_s^2 \sum_{j=1}^2 \sum_{i=1}^2 \left\| \m{D}_j^\dagger \m{H}_{ji}\m{E}_i - \m{I}_L \right\|_F^2 + \sigma^2 \sum_{j=1}^2 \|\m{D}_j\|_F^2.
\end{align}
We alternate between optimizing the encoders and the decoders.

\subsection{Encoder Optimization (Fixed Decoders)}
For fixed decoders $\{\m{D}_1, \m{D}_2\}$, the MSE objective decouples into two independent quadratically constrained quadratic programs (QCQPs), one for each transmitter. The subproblem for Tx-$1$ is:
\begin{subequations} \label{eq:p1}
\begin{align}
    \min_{\m{E}_1} \quad & \sum_{j=1}^2 \left\|\m{D}_j^\dagger \m{H}_{j1}\m{E}_1 - \m{I}_L\right\|_F^2 \tag{$P_1$} \\
    \text{s.t.} \quad & \|\m{E}_1\|_F^2 \le 1. \notag
\end{align}
\end{subequations}
Introducing the multiplier $\lambda_1 \ge 0$ for the power constraint gives the Lagrangian
\begin{align}
    \mathcal{L}(\m{E}_1, \lambda_1) &= \sum_{j=1}^2 \left\|\m{D}_j^\dagger \m{H}_{j1}\m{E}_1 - \m{I}_L\right\|_F^2 \notag \\
    &\quad + \lambda_1 \left(\|\m{E}_1\|_F^2 - 1\right),\label{eq:lagrangian_tx1}
\end{align}
and setting the derivative with respect to $\m{E}_1^*$ to zero yields
\begin{align}
    \bigg(\sum_{j=1}^2 \m{H}_{j1}^\dagger \m{D}_j \m{D}_j^\dagger &\m{H}_{j1} + \lambda_1 \m{I}_{M_1}\bigg)\m{E}_1 = \m{H}_{11}^\dagger \m{D}_1 + \m{H}_{21}^\dagger \m{D}_2. \label{eq:optimality_cond_tx1}
\end{align}
The optimal encoder $\m{E}_1^\star$ parameterized by $\lambda_1$ is:
\begin{align}\label{eq:E1_star}
   \m{E}_1^\star(\lambda_1) = \bigg( \sum_{j=1}^2 \m{H}_{j1}^\dagger \m{D}_j \m{D}_j^\dagger \m{H}_{j1} + \lambda_1 \m{I}_{M_1} \bigg)^{-1}\sum_{j=1}^2 \m{H}_{j1}^\dagger \m{D}_j.
\end{align}
If the unconstrained solution satisfies $\|\m{E}_1^\star(0)\|_F^2 \le 1$, we set $\lambda_1 = 0$. Otherwise, $\|\m{E}_1^\star(\lambda_1)\|_F^2$ decreases monotonically with $\lambda_1$, and the unique solution to $\|\m{E}_1^\star(\lambda_1)\|_F^2 = 1$ can be found efficiently by scalar bisection. The update for $\m{E}_2^\star$ is symmetric, obtained by substituting $(\m{H}_{11}, \m{H}_{21})$ with $(\m{H}_{12}, \m{H}_{22})$ in the derivations above.

\subsection{Decoder Optimization (Fixed Encoders)}
For fixed encoders $\{\m{E}_1, \m{E}_2\}$, the optimization over the decoders reduces to unconstrained strictly convex quadratic problems. The objective naturally decouples across the two receivers. For Rx-$1$, the subproblem is:
\begin{align} \label{eq:p3}
    \min_{\m{D}_1} \quad \sigma_s^2 \sum_{i=1}^2 \left\|\m{D}_1^\dagger \m{H}_{1i}\m{E}_i - \m{I}_L\right\|_F^2 + \sigma^2 \|\m{D}_1\|_F^2 \tag{$P_3$}.
\end{align}
Taking the derivative with respect to $\m{D}_1^*$ and equating it to zero yields the standard linear minimum mean-squared error (LMMSE) receiver:
\begin{align} \label{eq:D1_star}
    \m{D}_1^\star &= \bigg( \sum_{i=1}^2 \m{H}_{1i} \m{E}_i \m{E}_i^\dagger \m{H}_{1i}^\dagger + \frac{\sigma^2}{\sigma_s^2} \m{I}_{N_1} \bigg)^{-1}  \sum_{i=1}^2 \m{H}_{1i} \m{E}_i.
\end{align}
Similarly, the optimal decoder for Rx-$2$ is:
\begin{align} \label{eq:D2_star}
    \m{D}_2^\star &= \bigg( \sum_{i=1}^2 \m{H}_{2i} \m{E}_i \m{E}_i^\dagger \m{H}_{2i}^\dagger + \frac{\sigma^2}{\sigma_s^2} \m{I}_{N_2} \bigg)^{-1} \sum_{i=1}^2 \m{H}_{2i} \m{E}_i.
\end{align}

\subsection{Algorithm Analysis and Numerical Results}
Because each block update minimizes the objective over its block, the objective sequence is non-increasing and bounded below, hence convergent. \new{Moreover, since the two decoders decouple and are minimized jointly and exactly given the encoders, and vice versa, the iteration is an exact two-block alternating minimization over closed convex sets; every limit point of the iterates is therefore a stationary point \cite{grippoConvergenceBlockNonlinear2000}.} Global optimality is not guaranteed because the joint problem is non-convex. \new{We therefore use several random initializations and retain the solution with the smallest objective. In the SNR sweeps, one initialization is additionally set to the solution from the preceding SNR point (an \emph{SNR continuation} warm start); near the DoF boundary $L=\mathrm{ACDoF}^{\star}$ this prevents rare poor local minima at high SNR from dominating the average over channel realizations.}

\new{Per iteration, the decoder update at Rx-$j$ requires solving the $N_j\times N_j$ linear system in \eqref{eq:D1_star} ($O(N_j^3+N_j^2L)$). The encoder update \eqref{eq:E1_star} uses one eigendecomposition of its $M_i\times M_i$ Hermitian matrix, after which $\lambda_i$ is found by scalar bisection when needed; its leading cost is $O(M_i^3+M_i^2L)$. Thus the leading cost is cubic in the antenna dimensions and linear in $L$.}

\new{Table~\ref{tab:ao_convergence} reports the convergence behavior of AO for $L=5$ and $L=10$ with a maximum of $400$ iterations. To isolate the convergence cost at each SNR point, the table is generated without the SNR continuation: for each channel realization, AO is run five times from random initializations. Each run stops when the change between two successive objective values is below the tolerance $10^{-9}$ or at the $400$-iteration limit, and the run with the smallest final objective value is selected. Let $T\le400$ be the stopping iteration of this selected run and $J_t$ its objective after iteration $t$. We define $t_{\epsilon}\triangleq\min\{t:J_t\le(1+\epsilon)J_T\}$ and $\Delta_{50}\triangleq100[J_{\min\{50,T\}}/J_T-1]$. The reported entries are medians over $40$ channel realizations: $20$ Monte Carlo trials, each drawing two independent realizations, one per channel use.}

\begin{table}[!t]
\centering
\color{black}
\caption{\new{Empirical convergence of AO with at most $400$ iterations.}}
\label{tab:ao_convergence}
\small
\setlength{\tabcolsep}{3.5pt}
\renewcommand{\arraystretch}{1.0}
\begin{tabular}{c|ccc|ccc}
\hline
& \multicolumn{3}{c|}{$L=5$} & \multicolumn{3}{c}{$L=10$} \\
SNR (dB) & $t_{1\%}$ & $t_{0.1\%}$ & $\Delta_{50}$ (\%)
         & $t_{1\%}$ & $t_{0.1\%}$ & $\Delta_{50}$ (\%) \\
\hline
$0$  & $23.5$  & $50$    & $0.10$  & $8$     & $19$  & $0.00$ \\
$8$  & $87$    & $185$   & $3.38$  & $28$    & $57$  & $0.17$ \\
$16$ & $310.5$ & $387.5$ & $28.45$ & $97$    & $193$ & $4.46$ \\
$24$ & $380$   & $398$   & $79.55$ & $314.5$ & $389$ & $26.98$ \\
$32$ & $384.5$ & $399$   & $41.35$ & $371$   & $397$ & $40.99$ \\
$40$ & $361$   & $396.5$ & $11.33$ & $369.5$ & $397$ & $58.80$ \\
\hline
\end{tabular}
\end{table}

\new{Cold-start convergence slows as the SNR grows. At $L=5<\mathrm{ACDoF}^{\star}$ AirComp tasks, the median $t_{0.1\%}$ rises from $50$ iterations at $0$\,dB to essentially the full limit at $24$--$40$\,dB, with $\Delta_{50}$ peaking at $79.55\%$; the trend at $L=10$ is similar, with $\Delta_{50}$ rising to $58.80\%$ at $40$\,dB. In the SNR sweeps, the SNR continuation removes most of the convergence cost---the warm-started run starts near the final value from the slower SNR. These results motivate the $400$-iteration limit.}

\new{In each Monte Carlo trial, all four schemes use the same two independently generated channel realizations, one for each channel use, and in the SNR sweeps all four schemes employ the SNR continuation introduced above. Over the two uses, the two receivers compute $4L$ scalar sums in total; hence, the plotted per-computation MSE is the sum MSE divided by $4L$, averaged over the Monte Carlo trials. Unless noted otherwise, all channel coefficients are drawn independently from $\mathcal{CN}(0,1)$. Figure~\ref{fig:mseL_uneq} instead uses independent Rician coefficients. Let $X_{\mathrm R}$ denote a Rician coefficient with $K_{\mathrm R}=5$, normalized such that $\mathbb E|X_{\mathrm R}|^2=1$. The direct-link coefficients are drawn as $X_{\mathrm R}$ and the cross-link coefficients as $\sqrt{0.5}X_{\mathrm R}$.}

\new{The first baseline is \emph{AirComp-time-sharing}: the two receivers are served in turn, Rx-$1$ in the first channel use and Rx-$2$ in the second. In each channel use, the two transmitters and the receiver served in that use apply the single-receiver specialization of the same AirComp AO design: the two encoders and this receiver's LMMSE sum decoder are updated alternately to compute all $2L$ sums \cite{huRISAssistedOvertheAirFederated2022}. Thus each channel use is an AO-designed single-receiver MAC-AirComp link, while the overall baseline time-shares the desired computation across receivers. Remark~\ref{rem:mimo_mac} then requires $2L\le\min\{M_1,M_2,N_j\}$ for both receivers, motivating an upper bound for vanishing MSE}
\begin{align}\label{eq:aircomp_ts_dof}
\new{L \le \mathrm{ACDoF}_{\mathrm{TS}} \triangleq \tfrac12\min\{M_1,M_2,N_1,N_2\}.}
\end{align}

\new{The second baseline is \emph{BC-time-sharing}: Tx-$1$ sends $2L$ symbols in the first use, Tx-$2$ in the second, and each receiver adds its LMMSE estimates. Because each transmitter is silent in one of the two uses, its active-use power is set to $2$, giving the same unit average power over the two uses as the always-on schemes. For the active transmitter, the heuristic alternately updates its encoder and the two LMMSE decoders. During each encoder update, a larger weight is assigned to the receiver with the higher current MSE, thereby preventing the design from favoring either receiver. Applying the point-to-point bound $2L\le\min\{M_i,N_j\}$ to both $i,j\in\{1,2\}$ gives the same upper bound $\mathrm{ACDoF}_{\mathrm{TS}}$.}

\new{The third baseline is \emph{C-MAC source recovery}. Both transmitters are active in every channel use and, as in AO, Tx-$i$ transmits $\ve{x}_i=\m{E}_i\ve{s}_i$ with $\m{E}_i\in\mathbb C^{M_i\times L}$, $\|\m{E}_i\|_F^2\le 1$. The receivers' task differs: each receiver treats the channel use as a two-user MIMO MAC whose message set requires it to recover \emph{both} source vectors---a two-receiver MIMO compound MAC \cite{sankarErgodicFadingInterference2011}. Rx-$j$ therefore applies a wider decoder $\m{D}_j\in\mathbb C^{N_j\times 2L}$ that estimates the stacked sources and adds the two blocks,}
\begin{align}\label{eq:cmac_est}
\new{\begin{bmatrix}\hat{\ve{s}}_1^{(j)}\\ \hat{\ve{s}}_2^{(j)}\end{bmatrix}=\m{D}_j^\dagger\ve{y}_j,\qquad
\hat{\ve{f}}_j=\hat{\ve{s}}_1^{(j)}+\hat{\ve{s}}_2^{(j)},}
\end{align}
\new{where $\hat{\ve{s}}_i^{(j)}$ denotes the estimate of $\ve{s}_i$ at Rx-$j$. With the selection matrices $\m{S}_1=[\m{I}_L\ \m{0}]^\mathsf{T}$ and $\m{S}_2=[\m{0}\ \m{I}_L]^\mathsf{T}$, the design minimizes the source-recovery objective}
\begin{align}\label{eq:cmac_mse}
\new{J_{\mathrm{CMAC}}=\sum_{j=1}^2\left(\sigma_s^2\sum_{i=1}^2
\left\|\m{D}_j^\dagger\m{H}_{ji}\m{E}_i-\m{S}_i\right\|_F^2
+\sigma^2\|\m{D}_j\|_F^2\right),}
\end{align}
\new{which is \eqref{eq:MSE_expanded} with the common target $\m{I}_L$ replaced by the separated targets $\m{S}_i$. Only the objective changes; the alternating updates mirror those of AO. We use multiple random initializations---with the same SNR continuation as AO---and retain the smallest $J_{\mathrm{CMAC}}$; global optimality is not claimed. We plot the per-computation MSE of the sum estimate $\hat{\ve{f}}_j$ in \eqref{eq:cmac_est}; criterion \eqref{eq:cmac_mse} instead enforces recovery of both source vectors. Because each receiver must resolve the $2L$ independent streams ($2L\le N_j$) that the transmitters emit as $L$ streams each ($L\le M_i$), vanishing MSE requires}
\begin{align}\label{eq:cmac_dof}
\new{L\le \mathrm{ACDoF}_{\mathrm{CMAC}} \triangleq \min\left\{M_1,M_2,\tfrac{N_1}{2},\tfrac{N_2}{2}\right\}.}
\end{align}

\new{We evaluate the four schemes on the asymmetric network $(M_1,M_2;N_1,N_2)=(10,17;12,13)$. Theorem~\ref{thm:4} gives $\mathrm{ACDoF}^{\star}=\min\{10,17,12,13,\tfrac13\max\{27,25\}\}=9$, for which the minimum is attained by the term $(M_1+M_2)/3$ alone. The time-sharing threshold in \eqref{eq:aircomp_ts_dof} is $\mathrm{ACDoF}_{\mathrm{TS}}=5$, and the C-MAC threshold in \eqref{eq:cmac_dof} is $\mathrm{ACDoF}_{\mathrm{CMAC}}=\min\{10,17,6,6.5\}=6$. Thus the MSE-versus-$L$ plots mark $L=5$ and $L=9$ as the time-sharing and optimal AirComp reference values, respectively; the C-MAC transition at $L=6$ falls between them.}

\new{The three thresholds have different theoretical standings. $\mathrm{ACDoF}^{\star}$ is the exact AirComp DoF, supported by both an explicit achieving construction and a converse (Theorem~\ref{thm:4}). In contrast, $\mathrm{ACDoF}_{\mathrm{TS}}$ and $\mathrm{ACDoF}_{\mathrm{CMAC}}$ follow from necessary conditions only: they are upper bounds on the number of AirComp tasks $L$ at which the respective baseline classes can maintain vanishing MSE. Moreover, none of the implemented iterative algorithms---including AO---carries a theoretical guarantee of attaining its threshold. Empirically, however, each algorithm sustains the $O(1/\rho)$ decay up to its threshold and develops an MSE floor immediately beyond it, so all three thresholds appear tight for the implemented designs.}

\begin{figure}[!t]
  \centering
  \color{black}
  \includegraphics[width=4.3in]{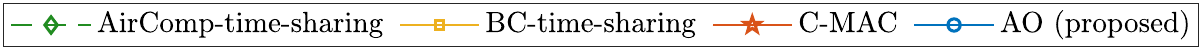}\\[1pt]
  \begin{subfigure}[t]{0.49\textwidth}
    \centering\includegraphics[width=0.82\linewidth]{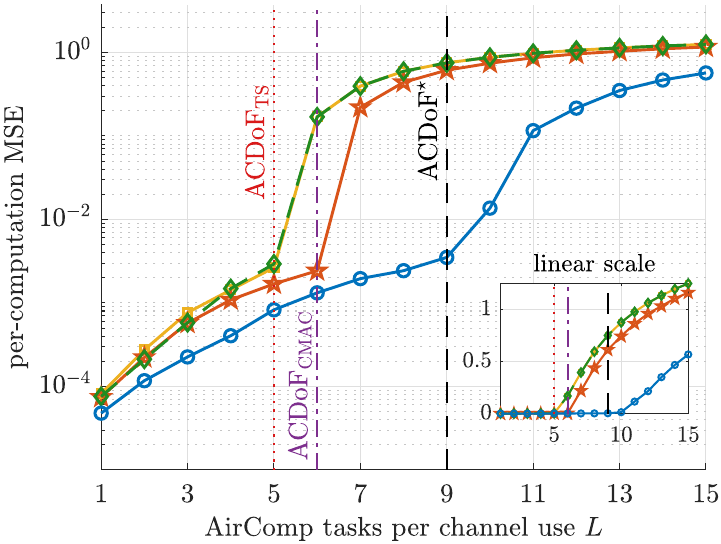}
    \caption{Rayleigh fading: all channel coefficients are i.i.d.\ $\mathcal{CN}(0,1)$.}\label{fig:mseL_eq}
  \end{subfigure}\hfill
  \begin{subfigure}[t]{0.49\textwidth}
    \centering\includegraphics[width=0.82\linewidth]{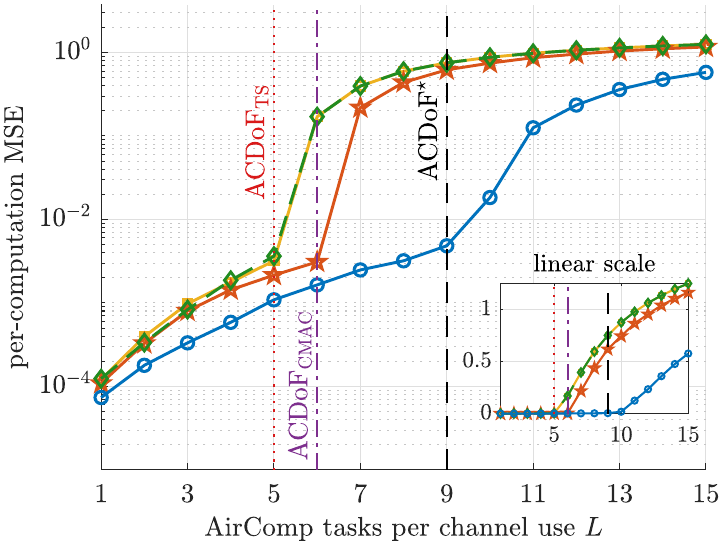}
    \caption{Rician fading: direct-link coefficients are drawn as $X_{\mathrm R}$ and cross-link coefficients as $\sqrt{0.5}X_{\mathrm R}$, where $K_{\mathrm R}=5$ and $\mathbb E|X_{\mathrm R}|^2=1$.}\label{fig:mseL_uneq}
  \end{subfigure}
  \caption{\new{Per-computation MSE versus $L$ for the $(10,17;12,13)$ network at $30$\,dB. The vertical lines mark $\mathrm{ACDoF}_{\mathrm{TS}}=5$, $\mathrm{ACDoF}_{\mathrm{CMAC}}=6$, and $\mathrm{ACDoF}^{\star}=9$. Main axes are logarithmic; the insets are linear.}}
  \label{fig:mseL}
\end{figure}

\new{In both fading models, the four curves in Figure~\ref{fig:mseL} separate exactly at their respective thresholds. The two time-sharing curves increase sharply from $L=5$ to $L=6$. The C-MAC curve follows one step later. AO remains low through $L=9$ and increases from $L=10$.}

\begin{figure}[t]
  \centering
  \color{black}
  \includegraphics[width=4.3in]{snrLegend.pdf}\\[1pt]
  \begin{subfigure}[t]{0.29\linewidth}
    \centering\includegraphics[width=\linewidth]{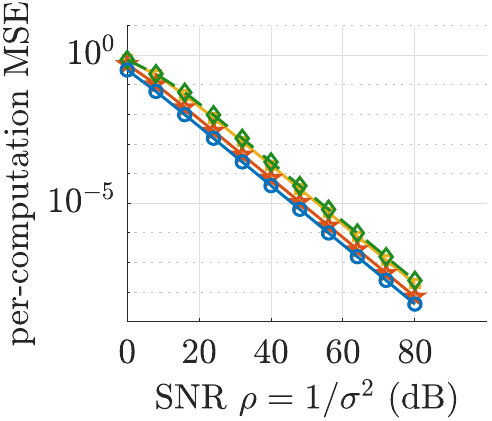}
    \caption{$L=5=\mathrm{ACDoF}_{\mathrm{TS}}$.}\label{fig:snr_l5}
  \end{subfigure}\hspace{0.025\linewidth}
  \begin{subfigure}[t]{0.29\linewidth}
    \centering\includegraphics[width=\linewidth]{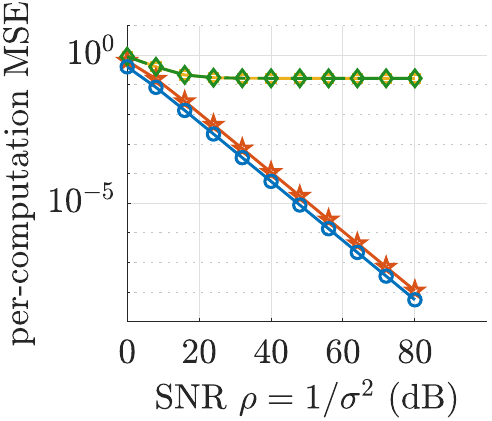}
    \caption{$L=6=\mathrm{ACDoF}_{\mathrm{CMAC}}$.}\label{fig:snr_l6}
  \end{subfigure}\hspace{0.025\linewidth}
  \begin{subfigure}[t]{0.29\linewidth}
    \centering\includegraphics[width=\linewidth]{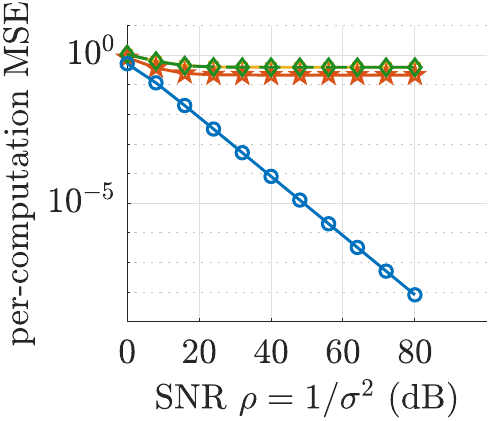}
    \caption{$L=7>\mathrm{ACDoF}_{\mathrm{CMAC}}$.}\label{fig:snr_l7}
  \end{subfigure}\par\medskip
  \begin{subfigure}[t]{0.29\linewidth}
    \centering\includegraphics[width=\linewidth]{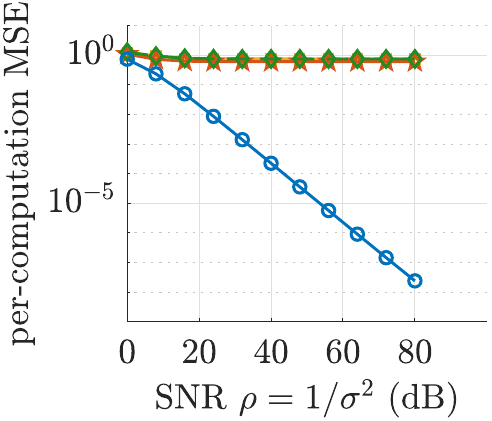}
    \caption{$L=9=\mathrm{ACDoF}^{\star}$.}\label{fig:snr_l9}
  \end{subfigure}\hspace{0.025\linewidth}
  \begin{subfigure}[t]{0.29\linewidth}
    \centering\includegraphics[width=\linewidth]{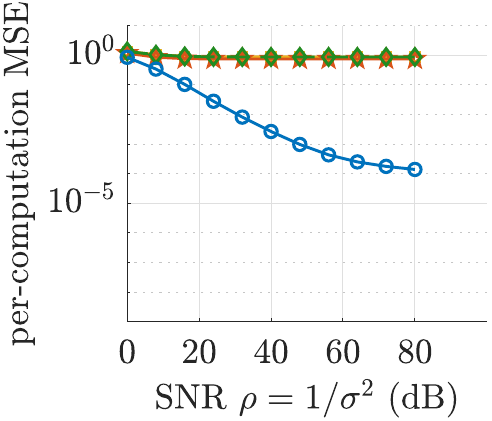}
    \caption{$L=10>\mathrm{ACDoF}^{\star}$.}\label{fig:snr_l10}
  \end{subfigure}
  \caption{\new{Per-computation MSE versus SNR for the $(10,17;12,13)$ network at selected numbers of AirComp tasks around the time-sharing, C-MAC, and optimal AirComp thresholds. All panels span $0$--$80$\,dB.}}
  \label{fig:mseSNR}
\end{figure}

\new{Figure~\ref{fig:mseSNR} makes the distinction sharper. At $L=5$, every scheme has high-SNR slope approximately one on the log-MSE scale, consistent with $O(1/\rho)$. At $L=6$, AirComp-time-sharing and BC-time-sharing approach a common MSE floor, whereas AO and C-MAC retain slope one through $80$\,dB. At $L=7$, C-MAC develops an MSE floor as well, while AO alone retains slope one. At $L=9=\mathrm{ACDoF}^{\star}$, only AO continues to decrease with slope approximately one, while the baselines approach nonzero MSE floors. At $L=10$, every baseline exhibits an MSE floor.}

\new{In summary, the finite-SNR MSE of every implemented scheme transitions at its theoretical threshold---$\mathrm{ACDoF}_{\mathrm{TS}}=5$, $\mathrm{ACDoF}_{\mathrm{CMAC}}=6$, and $\mathrm{ACDoF}^{\star}=9$---and only AO sustains the $O(1/\rho)$ decay up to the network AirComp DoF. These observations support the operational relevance of the AirComp DoF characterization at finite SNR, while remaining empirical: they neither prove that the implemented algorithms attain their thresholds nor guarantee the same behavior for other channel realizations. An algorithm that provably attains the AirComp DoF and performs well under a finite power budget remains to be designed.}

\Needspace{5\baselineskip}
\section{\texorpdfstring{\new{Discussion}}{Discussion}}\label{sec:generalization}

\new{Three extensions remain open. First, consider a symmetric SISO network with $K$ transmitters and $K$ receivers, all demanding $s_1+\cdots+s_K$. The converse of Section~\ref{sec:5} extends to
\begin{align}
\mathrm{ACDoF} \le \frac{K}{2K-1},
\end{align}
whereas, for arbitrary $K\ge3$, the only achievable value established here is the time-sharing value $1/K$; the achievability construction of Theorem~\ref{thm:1} applies only to $K=2$. Closing this gap remains open.}

\new{Second, suppose Rx-$j$ demands a different linear target $f_j=a_{j1}s_1+a_{j2}s_2$, with nonzero $a_{ji}$. The AirComp DoF in Theorem~\ref{thm:4} remains achievable. Specifically, design the linear transceivers for the common sum over the scaled links $\widetilde{\m{H}}_{ji}=a_{ji}^{-1}\m{H}_{ji}$ and apply the same transceivers over the actual links; the resulting effective coefficient of source $i$ at Rx-$j$ is $a_{ji}$. Fixed nonzero scalings preserve genericity. The converse does not extend because it relies on both receivers demanding the same function. For linearly independent targets, the two requested functions already span the two-source space, leaving no additional independent $K$-dimensional source combination for the $3K$ invertibility argument. The resulting argument gives only the weaker cooperative upper bound $\tfrac12\max\{M_1+M_2,N_1+N_2\}$, together with the individual antenna bounds. Whether heterogeneous targets yield a strictly larger AirComp DoF is open.}

\new{Third, rank-deficient channel matrices provide a structured extension beyond the full-rank generic model considered here. Specifically, the channel from Tx-$i$ to Rx-$j$ can be modeled as
\begin{align}
    \m{H}_{ji}=\m{H}_{ji,A}\m{H}_{ji,B},
\end{align}
where $\m{H}_{ji,A}\in\mathbb{C}^{N_j\times r_{ji}}$ and $\m{H}_{ji,B}\in\mathbb{C}^{r_{ji}\times M_i}$ are independently generated generic matrices, with $r_{ji}<\min\{M_i,N_j\}$. The resulting matrix $\m{H}_{ji}$ has rank $r_{ji}$ almost surely. Characterizing the AirComp DoF under such rank constraints requires new achievability and converse arguments and remains a direction for future work.}

\section{Conclusion}

This paper characterized the AirComp DoF of two-transmitter, two-receiver networks for both SISO channels and generic MIMO settings. For the SISO case, the characterization covers time-varying and time-invariant channels under explicit channel conditions. For the generic MIMO case, we developed a unified achievability framework based on \new{virtual subnetwork packing, in which an explicitly constructed realization of the channel matrices decomposes into the packed subnetworks and makes the relevant determinant polynomials nonzero, thereby establishing the required full-rank conditions almost surely for generic channels}. The converse results for both the SISO and MIMO settings were established via information-theoretic MSE--equivocation arguments combined with network invertibility analysis.

\new{The cooperative upper bound in Theorem~\ref{thm:4} reflects the three-dimensional signal-space accounting proved in Section~\ref{proof:MIMOconverse}: two dimensions deliver the desired sum to the two receivers, while network invertibility forces the joint observations to retain one additional independent source dimension. This places the generic SISO value $2/3$ between separate-service time sharing ($1/2$) and the single-receiver ideal sum channel ($1$), and yields the MIMO bound $\tfrac13\max\{M_1+M_2,N_1+N_2\}$.}

In addition to the asymptotic AirComp DoF characterization, we proposed a linear transceiver design for finite-SNR operation based on alternating minimization of the network MSE\@. \new{Numerical results compared this design with AirComp-time-sharing, BC-time-sharing, and C-MAC source recovery under the same channel and power conditions.} These results provide a unified perspective on AirComp over two-transmitter, two-receiver networks.

\section*{Acknowledgment}
The authors used Google Gemini, Anthropic Claude, and OpenAI Codex to assist with language editing and \LaTeX\ formatting throughout Sections~I--VII, the Conclusion, and the Appendices, and with the optimization and simulation code in Section~VI. The authors reviewed and edited all AI-assisted content, independently verified all analytical derivations, code, and simulation results, and take full responsibility for the final publication.

\appendices
\section{AirComp DoF of the SISO Multiple Access Channel} \label{appendix_sisomac}

As a baseline, we evaluate the conventional SISO MAC\@. This formulation is a degenerate case of the general network model presented in \eqref{eq:modely1}, obtained by setting $\ve{y}_1 = \ve{y}_2 = \ve{y}$ and $\ve{z}_1 = \ve{z}_2 = \ve{z}$.

\begin{lemma}\label{lemma:SISOMAC}
For any given SISO MAC realization $\mathscr{H}$, the AirComp DoF is exactly one; that is, $\mathrm{ACDoF}(\mathscr{H}) = 1$.
\end{lemma}

\begin{proof}
\textit{Converse:} Assume that for any SNR $\rho$, there exists a computation code $\mathfrak{C}_{\rho}$ of length $N_{\rho}$ encoding $K_{\rho}$ symbols such that $\mathrm{MSE}(\mathfrak{C}_{\rho}, \mathscr{H}) = O(1/\rho)$. Let $d = \lim_{\rho \to \infty} K_{\rho}/N_{\rho}$ denote the achieved AirComp DoF.

Applying the MSE--equivocation bound (detailed in Appendix~\ref{appendix_mse}) to the $K_\rho$ computation symbols, the mutual information between the target computation $f^{K_\rho}$ and the received signal $\ve{y}^{N_\rho}$ is lower-bounded by:
\begin{align}
    I(f^{K_\rho}; \ve{y}^{N_\rho}) &= h(f^{K_\rho}) - h(f^{K_\rho} | \ve{y}^{N_\rho}) \label{eq:mac_1} \\
    &\ge h(f^{K_\rho}) + K_\rho \log\rho + K_\rho o(\log \rho) \label{eq:mac_2} \\
    &= K_\rho \log\rho + K_\rho o(\log \rho). \label{eq:mac_3}
\end{align}
Note that \eqref{eq:mac_3} holds because the unconditional differential entropy satisfies $h(f^{K_\rho}) = K_\rho h(f) = K_\rho o(\log \rho)$.

To establish an upper bound on $I(f^{K_\rho}; \ve{y}^{N_\rho})$, we evaluate the differential entropy of the received signal:
\begin{align} \label{eq:mac_4}
I(f^{K_\rho}; \ve{y}^{N_\rho}) &= h(\ve{y}^{N_\rho}) - h(\ve{y}^{N_\rho} | f^{K_\rho}).
\end{align}
The first term is constrained by the average power limit $\mathbb{E}[\lVert \ve{y} \rVert^2] \le 2M_o^2 + \sigma^2$. Because a Gaussian distribution maximizes differential entropy for a specified variance, we have $h(\ve{y}^{N_\rho}) \le N_\rho\log(2M_o^2 + 1/\rho) + N_\rho o(\log\rho)$. For the conditional entropy term, conditioning on the independent message $s_2^{K_\rho}$ yields:
\begin{align} \label{eq:mac_entropy_cond}
    h(\ve{y}^{N_\rho} | f^{K_\rho}) &\ge h(\ve{y}^{N_\rho} | f^{K_\rho}, s_2^{K_\rho})  \\
    &= h(\ve{z}^{N_\rho} | f^{K_\rho}, s_2^{K_\rho}) = h(\ve{z}^{N_\rho})  \\
    &= -N_\rho\log\rho + N_\rho o(\log\rho).
\end{align}
Substituting these bounds into \eqref{eq:mac_4} and combining with \eqref{eq:mac_3} produces the unified inequality:
\begin{align}
    K_\rho \log\rho + K_\rho o(\log \rho)
    &\le N_\rho \log \left(2M_o^2 + \frac{1}{\rho} \right) \notag\\
    &\quad + N_\rho\log\rho + N_\rho o(\log \rho).
\end{align}
Dividing this inequality by $N_\rho\log\rho$ and taking $\rho \to \infty$ gives $d\le1$: the term involving $K_\rho o(\log\rho)$ vanishes because $K_\rho/N_\rho$ has a finite limit, and the remaining lower-order terms vanish similarly.

\textit{Achievability:} Consider a single-shot $(1,1)$ code with the zero-forcing strategy:
\begin{align}
    \ve{x}_1(\nu) = \frac{\mu}{H_1(\nu)} s_1(\nu), \quad \ve{x}_2(\nu) = \frac{\mu}{H_2(\nu)} s_2(\nu),
\end{align}
where $H_i(\nu)$ denotes the scalar channel coefficient from Tx-$i$ to the common receiver, and the receiver estimates the sum computation via $\hat{f} = \ve{y}/\mu$. By setting the scaling factor to $\mu = 1 / (M_o \sigma_s)$, we strictly satisfy the transmit power constraints. Consequently, $\mathrm{MSE}(\mathfrak{C}, \mathscr{H}) = \sigma^2 / \mu^2 = O(1/\rho)$, so an AirComp DoF of $1$ is achievable, concluding the proof.
\end{proof}

\new{\begin{remark}[MIMO MAC]\label{rem:mimo_mac}
Lemma~\ref{lemma:SISOMAC} extends to the single-receiver MIMO MAC, where the transmitters have $M_1,M_2$ antennas and the common receiver has $N$ antennas: for a generic channel, the AirComp DoF is $\min\{M_1,M_2,N\}$ almost surely. For achievability, linear precoding and combining reduce the generic MIMO MAC to $\min\{M_1,M_2,N\}$ parallel scalar MACs, each of AirComp DoF $1$ by Lemma~\ref{lemma:SISOMAC}. For the converse, since $f^K$ and $s_2^K$ together determine $s_1^K$, a genie revealing $s_2^K$ to the receiver reduces the problem to conveying $s_1^K$ over a point-to-point $M_1\times N$ MIMO channel, whose vanishing-MSE DoF is at most $\min\{M_1,N\}$ (by the converse steps of Lemma~\ref{lemma:SISOMAC} with $(\ve{y}^N, s_2^K)$ in place of $\ve{y}^N$); hence the AirComp DoF is at most $\min\{M_1,N\}$ almost surely, and $\min\{M_2,N\}$ follows by symmetry. This provides the single-receiver baseline $\min\{M_1,M_2,N_j\}$ used by the time-sharing schemes of Section~\ref{sec:numerical}.
\end{remark}}

\section{Proof of Lemma~\ref{lem:achievability_packing}} \label{appendix_sisomacchievability}
\new{\emph{Outline.} Before the details, we summarize the three steps of the proof. \emph{Step 1 (construction):} over $N=3$ channel uses, each transmitter applies a block precoder and each receiver a block combiner. The combiners satisfy the null-space alignment condition \eqref{eq:chooseV}, which makes the two receivers observe the \emph{same} effective MAC, \eqref{eq:fe2}. The target dimension is $K=2a+3b+3c$. Under the explicitly constructed realization of the channel matrices used in Step~3, the packed copies occupy parallel, noninteracting subspaces: each $(1,1;1,1)$ copy runs the SISO scheme of Theorem~\ref{thm:1} and contributes two computations over the three channel uses, while each $(1,1;2,1)$ or $(1,1;1,2)$ copy contributes one computation per channel use, for a total AirComp DoF of $\frac{2}{3}a+b+c$. \emph{Step 2 (when the construction is feasible):} given the packing constraints and null-space construction, feasibility reduces to four nonzero determinants formed from the channel and design matrices. Each determinant, defined in \eqref{eq:deffji}, is a rational polynomial in its variables. \emph{Step 3 (the requirements hold almost surely):} we evaluate the determinants under one explicitly constructed realization of the channel matrices that decomposes exactly into the packed subnetworks specified in items~1)--3) and \eqref{eq:global_H_assignment}, and verify that they are nonzero for this realization. Since a rational polynomial that is nonzero at a single point is not identically zero, all four determinants are nonzero almost surely for generic channel matrices. The details follow.}

\begin{table}[H]
\centering
\color{black}
\caption{\new{Notation for the virtual subnetwork packing.}}
\label{tab:packing_notation}
\begin{tabular}{@{}l p{0.56\columnwidth}@{}}
\hline
$(m_1,m_2;n_1,n_2)$ & subnetwork type: $m_i$ antennas at Tx-$i$ and $n_j$ antennas at Rx-$j$\\
$k(t)$ & number of packed copies of subnetwork type $t$ in \eqref{eq:packingconstraint}\\
$a$ & copies of the $(1,1;1,1)$ subnetwork, each contributing AirComp DoF $2/3$\\
$b$ & copies of the $(1,1;2,1)$ subnetwork, each contributing AirComp DoF $1$\\
$c$ & copies of the $(1,1;1,2)$ subnetwork, each contributing AirComp DoF $1$\\
$\m{U}_i\in\mathbb{C}^{NM_i\times K}$ & block precoder at Tx-$i$ over the $N$-symbol extension\\
$\m{V}_j\in\mathbb{C}^{NN_j\times K}$ & block combiner at Rx-$j$ over the $N$-symbol extension\\
$\m{H}_{ji}^{[N]}$ & block-diagonal extended channel $\mathrm{diag}(\m{H}_{ji}(1),\ldots,\m{H}_{ji}(N))$\\
$\m{Q}$ & auxiliary matrix selecting the combiners through the null-space alignment condition \eqref{eq:chooseV}\\
$f_{ji}$ & full-rank certificates $\det(\m{V}_j^\dagger\m{H}_{ji}^{[N]}\m{U}_i)$ in \eqref{eq:deffji}; the scheme is feasible whenever all four are nonzero\\
\hline
\end{tabular}
\end{table}

Choose $N=3$ channel uses, and set $K/N = (2/3)a+b+c$,
so that $K=2a+3b+3c$ is an integer. Let $\m{U}_i\in\mathbb{C}^{NM_i\times K}$ and $\m{V}_i\in\mathbb{C}^{NN_i\times K}$ denote the precoding and combining matrices at Tx-$i$ and Rx-$i$, respectively, for $i\in\{1,2\}$, and let $\ve{x}_i^{(N)} = \m{U}_i \tilde{\ve{x}}_i^{(K)}$, $\tilde{\ve{y}}_i^{(K)} = \m{V}_i^\dagger \ve{y}_i^{(N)}$,
where $\tilde{\ve{x}}_i^{(K)}\in\mathbb{C}^K$ and $\tilde{\ve{y}}_i^{(K)}\in\mathbb{C}^K$ are the effective inputs and outputs. The effective outputs can be expressed as
\begin{align} \label{eq:MACs}
    \tilde{\ve{y}}_i^{(K)}
    &= \underbrace{\left(\m{V}_i^\dagger \m{H}_{i1}^{[N]} \m{U}_1\right)}_{\tilde{\m{H}}_{i1}} \tilde{\ve{x}}_1^{(K)} \notag \\
    &\quad + \underbrace{\left(\m{V}_i^\dagger \m{H}_{i2}^{[N]} \m{U}_2\right)}_{\tilde{\m{H}}_{i2}} \tilde{\ve{x}}_2^{(K)} + \underbrace{\m{V}_i^\dagger \ve{z}_i^{(N)}}_{\tilde{\ve{z}}_i}.
\end{align}
Thus, each receiver observes a $K$-antenna MAC with effective channel matrices
\begin{align}
    \tilde{\m{H}}_{ij} \triangleq \m{V}_i^\dagger \m{H}_{ij}^{[N]} \m{U}_j \in \mathbb{C}^{K\times K}.
\end{align}
Next define
\begin{align}\label{eq:defHU}
    \m{H}_U \triangleq
    \begin{bsmallmatrix}
        \m{H}_{11}^{[N]} & \m{H}_{12}^{[N]} \\
        \m{H}_{21}^{[N]} & \m{H}_{22}^{[N]}
    \end{bsmallmatrix}
    \begin{bsmallmatrix}
        \m{U}_1 & \\
        & \m{U}_2
    \end{bsmallmatrix}
    \in \mathbb{C}^{N(N_1+N_2)\times 2K}.
\end{align}
Since $\m{H}_U$ is a tall generic matrix, it has full column rank almost surely, so $\m{H}_U^\dagger\m{H}_U$ is invertible almost surely. Our goal is to enforce
\begin{align}
    \tilde{\m{H}}_{11}=\tilde{\m{H}}_{21},
    \qquad
    \tilde{\m{H}}_{12}=\tilde{\m{H}}_{22},
\end{align}
so that the two effective MACs are identical up to noise, while simultaneously ensuring that each effective channel matrix has full rank $K$. To this end, let us define the combining matrices according to
\begin{align} \label{eq:chooseV}
    &\begin{bmatrix} \m{V}_1^\dagger & -\m{V}_2^\dagger \end{bmatrix} \notag \\
    &= \m{Q}\bigl(\m{I}_{\bar{N}}
    -\m{H}_U(\m{H}_U^\dagger\m{H}_U)^{-1}\m{H}_U^\dagger\bigr).
\end{align}
for some $\m{Q}\in\mathbb{C}^{K\times N(N_1+N_2)}$. From \eqref{eq:defHU} and \eqref{eq:chooseV}, we obtain
\begin{align}
    \begin{bmatrix}
        \m{V}_1^\dagger & -\m{V}_2^\dagger
    \end{bmatrix}\m{H}_U = \m{0},
\end{align}
which implies
\begin{align}
    \m{V}_1^\dagger \m{H}_{11}^{[N]} \m{U}_1
    = \m{V}_2^\dagger \m{H}_{21}^{[N]} \m{U}_1, &&
    \m{V}_1^\dagger \m{H}_{12}^{[N]} \m{U}_2
    = \m{V}_2^\dagger \m{H}_{22}^{[N]} \m{U}_2.
    \label{eq:fe2}
\end{align}
Hence the two effective MACs are identical.

It remains to verify the full-rank property. For each $(j,i)\in\{1,2\}^2$, define
\begin{align}\label{eq:deffji}
    f_{ji} \triangleq \det\!\bigl(\m{V}_j^\dagger \m{H}_{ji}^{[N]} \m{U}_i\bigr).
\end{align}
Recall that for an $N$-symbol extension, the channel matrices $\m{H}_{ji}^{[N]} = \mathrm{diag}(\m{H}_{ji}(1), \dots, \m{H}_{ji}(N))$ are inherently block-diagonal. Therefore, $f_{ji}$ is viewed as a rational polynomial in the independent entries of the individual channel matrices $\m{H}_{ji}(n)$ for $n \in \{1,\ldots,N\}$, as well as the entries of $\m{U}_1$, $\m{U}_2$, and $\m{Q}$.

It suffices to exhibit one deterministic \emph{assignment}---a choice of the independent variables $\m{H}_{ji}(n)$ (which constructs $\m{H}_{ji}^{[N]}$), $\m{U}_i$, and $\m{Q}$, with the combiners $\m{V}_j$ then uniquely determined by the null-space alignment condition \eqref{eq:chooseV}---for which all four determinants are nonzero. Each $f_{ji}$ is then a nonzero rational polynomial, and because the actual channel coefficients across the $N$ channel uses are independent and continuously distributed, $f_{ji}\neq 0$ almost surely.
\begin{enumerate}
\item For each $(1,1;2,1;\mathcal{H})$ network (let us identify this case by including a $(2,1)$ in the subscript) over each channel use $\nu\in\{1,2,3\}$, choose the same assignment $\m{H}_{11,(2,1)} = \begin{bmatrix} 1 & 0 \end{bmatrix}^T$,    $\m{H}_{12,(2,1)} = \begin{bmatrix} 0 & 1 \end{bmatrix}^T$,    $\m{H}_{21,(2,1)} = \m{H}_{22,(2,1)} = \m{U}_{1,(2,1)}=\m{U}_{2,(2,1)}=1, \m{Q}_{(2,1)}=\begin{bmatrix}1&1&-1\end{bmatrix}$ which, according to \eqref{eq:chooseV} produces $\m{V}_{1,(2,1)}=\begin{bmatrix}1&1\end{bmatrix}$ and $\m{V}_{2,(2,1)}=1$, and $f_{ji,(2,1)}=1$ for all $i,j\in\{1,2\}$.
\item For each $(1,1;1,2;\mathcal{H})$ network (identified by including a $(1,2)$ in the subscript), over each channel use choose the assignment $\m{H}_{21,(1,2)} = \begin{bmatrix} 1 & 0 \end{bmatrix}^T,     \m{H}_{22,(1,2)} = \begin{bmatrix} 0 & 1 \end{bmatrix}^T,     \m{H}_{11,(1,2)} = \m{H}_{12,(1,2)} = \m{U}_{1,(1,2)}=\m{U}_{2,(1,2)}=1, \m{Q}_{(1,2)}=\begin{bmatrix}1&-1&-1\end{bmatrix}$ which, according to \eqref{eq:chooseV} produces $\m{V}_{2,(1,2)}=\begin{bmatrix}1&1\end{bmatrix}$ and $\m{V}_{1,(1,2)}=1$, and $f_{ji,(1,2)}=1$ for all $i,j\in\{1,2\}$.
\item For each $(1,1;1,1;\mathcal{H})$ (SISO) subnetwork choose the assignments over $N=3$ channel uses as,
\begin{align}
    \m{H}_{\mathrm{SISO}}(1) &= \begin{bsmallmatrix}-1&1\\1&1\end{bsmallmatrix}, \m{H}_{\mathrm{SISO}}(2) = \begin{bsmallmatrix}1&1\\1&-1\end{bsmallmatrix}, \notag \\
    \m{H}_{\mathrm{SISO}}(3) &= \begin{bsmallmatrix}1&0\\0&1\end{bsmallmatrix}, \notag \\
    \m{U}_{1,\mathrm{SISO}} &= \begin{bsmallmatrix}1&0\\ 0&1\\ 1&0\end{bsmallmatrix}, \m{U}_{2,\mathrm{SISO}} = \begin{bsmallmatrix}1&0\\ 0&1\\ 0&1\end{bsmallmatrix}, \notag \\
    \m{Q}_{\mathrm{SISO}} &= \begin{bsmallmatrix}1 & 0 & 2&-1&0&0\\ 0&1&0&0&-1&-2\end{bsmallmatrix} \notag \\
    \implies \m{V}_{1,\mathrm{SISO}} &= \begin{bsmallmatrix}1&0\\ 0&1\\ 2&0\end{bsmallmatrix}, \m{V}_{2,\mathrm{SISO}} = \begin{bsmallmatrix}1&0\\ 0&1\\ 0&2\end{bsmallmatrix}, \notag \\
    f_{ji,\mathrm{SISO}} &\stackrel{\eqref{eq:deffji}}{=} \det(\m{I}_2)=1, \;\; \forall i,j\in\{1,2\}. \label{eq:siso_assignments}
\end{align}
\end{enumerate}
Next, define the partial block-diagonal operator.

\begin{definition}[Partial Block-Diagonal Operator]
For matrices $\m{A}_i \in \mathbb{C}^{m_i\times n_i}$, $i\in[1:k]$, and integers
$M \ge \sum_{i=1}^k m_i$ and $N \ge \sum_{i=1}^k n_i$, define the \emph{partial block-diagonal operator}
\begin{align}
    \mathrm{pdiag}_{M,N}(\m{A}_1,\dots,\m{A}_k)\in\mathbb{C}^{M\times N}
\end{align}
as the matrix obtained by placing $\m{A}_1,\dots,\m{A}_k$ consecutively along the main block diagonal starting from the upper-left corner, with all remaining entries set to zero.
\end{definition}

Using this notation, we construct the global assignments as
\begin{align} \label{eq:global_H_assignment}
    &\m{H}_{ji}(n) = \mathrm{pdiag}_{N_j,M_i} \Big(\underbrace{\m{H}_{ji,\mathrm{SISO}}(n),\ldots,\m{H}_{ji,\mathrm{SISO}}(n)}_{a\ \text{copies}}, \notag \\
    &\quad \underbrace{\m{H}_{ji,(2,1)}(n),\ldots,\m{H}_{ji,(2,1)}(n)}_{b\ \text{copies}}, \notag\\
    &\quad \underbrace{\m{H}_{ji,(1,2)}(n),\ldots,\m{H}_{ji,(1,2)}(n)}_{c\ \text{copies}} \Big).
\end{align}
Note that the column space of both $\m{H}_{1i}^{[N]}$ and $\m{H}_{2i}^{[N]}$ corresponds to the transmit antennas of Tx-$i$ across the $N$ channel uses. Thus, a single column permutation matrix $\m{\Pi}_{c,i}\in\mathbb{R}^{NM_i\times NM_i}$ representing a reordering of Tx-$i$'s input dimensions naturally applies to both channel matrices simultaneously. Similarly, the row space of both $\m{H}_{j1}^{[N]}$ and $\m{H}_{j2}^{[N]}$ corresponds to the receive antennas of Rx-$j$, meaning a single row permutation matrix $\m{\Pi}_{r,j}\in\mathbb{R}^{NN_j\times NN_j}$ simultaneously acts on the outputs of both channels. Therefore, by applying suitable column permutation matrices $\m{\Pi}_{c,i}\in\mathbb{R}^{NM_i\times NM_i}$ at Tx-$i$ and row permutation matrices $\m{\Pi}_{r,j}\in\mathbb{R}^{NN_j\times NN_j}$ at Rx-$j$ for $i,j \in \{1,2\}$, let us  reorder the antennas so that
\begin{align}
    \m{\Pi}_{r,j}^T &\m{H}_{ji}^{[N]} \m{\Pi}_{c,i} = \mathrm{pdiag}_{NN_j,\,NM_i} \Big( \underbrace{\m{H}_{ji,\mathrm{SISO}}^{[N]},\dots,\m{H}_{ji,\mathrm{SISO}}^{[N]}}_{a \text{ copies}}, \notag \\
    &\quad \underbrace{\m{H}_{ji,(2,1)}^{[N]},\dots,\m{H}_{ji,(2,1)}^{[N]}}_{b \text{ copies}},  \underbrace{\m{H}_{ji,(1,2)}^{[N]},\dots,\m{H}_{ji,(1,2)}^{[N]}}_{c \text{ copies}} \Big). \label{eq:H_reordered}
\end{align}
Now we choose $\m{U}_i, \m{V}_j$ as follows.
\begin{align} \label{eq:precoder_assignment}
    \m{U}_i &= \m{\Pi}_{c,i} \, \mathrm{pdiag}_{NM_i,\,K} \Big(\underbrace{\m{U}_{i,\mathrm{SISO}},\ldots,\m{U}_{i,\mathrm{SISO}}}_{a\ \text{copies}}, \notag \\
    &\quad \underbrace{\m{U}_{i,(2,1)},\ldots,\m{U}_{i,(2,1)}}_{b\ \text{copies}},  \underbrace{\m{U}_{i,(1,2)},\ldots,\m{U}_{i,(1,2)}}_{c\ \text{copies}} \Big)\\
    \label{eq:combiner_assignment}
    \m{V}_j &= \m{\Pi}_{r,j} \, \mathrm{pdiag}_{NN_j,\,K} \Big( \underbrace{\m{V}_{j,\mathrm{SISO}},\ldots,\m{V}_{j,\mathrm{SISO}}}_{a\ \text{copies}}, \notag \\
    &\quad \underbrace{\m{V}_{j,(2,1)},\ldots,\m{V}_{j,(2,1)}}_{b\ \text{copies}},  \underbrace{\m{V}_{j,(1,2)},\ldots,\m{V}_{j,(1,2)}}_{c\ \text{copies}} \Big).
\end{align}
Under these choices, the orthogonal permutation matrices cancel in the product, and since the column dimensions of the blocks sum to $2a + 3b + 3c = K$, the product reduces to a standard block-diagonal matrix:
\begin{align} \label{eq:effective_channel_diag}
    \m{V}_j^\dagger \m{H}_{ji}^{[N]} \m{U}_i &= \mathrm{diag} \Big( \underbrace{\m{V}_{j,\mathrm{SISO}}^\dagger \m{H}_{ji,\mathrm{SISO}}^{[N]} \m{U}_{i,\mathrm{SISO}},\ldots}_{a\ \text{copies}}, \notag \\
    &\quad \underbrace{\m{V}_{j,(2,1)}^\dagger \m{H}_{ji,(2,1)}^{[N]} \m{U}_{i,(2,1)},\ldots}_{b\ \text{copies}}, \notag \\
    &\quad \underbrace{\m{V}_{j,(1,2)}^\dagger \m{H}_{ji,(1,2)}^{[N]} \m{U}_{i,(1,2)},\ldots}_{c\ \text{copies}} \Big).
\end{align}
Therefore,
\begin{align}
    \m{V}_1^\dagger \m{H}_{11}^{[N]} \m{U}_1
    &=
    \m{V}_2^\dagger \m{H}_{21}^{[N]} \m{U}_1, \label{eq:null1}\\
    \m{V}_1^\dagger \m{H}_{12}^{[N]} \m{U}_2
    &=
    \m{V}_2^\dagger \m{H}_{22}^{[N]} \m{U}_2,
    \label{eq:null2}
\end{align}
and
\begin{align}
    \rank{\m{V}_j^\dagger \m{H}_{ji}^{[N]} \m{U}_i}
    &=
    a\,\rank{\m{V}_{j,\mathrm{SISO}}^\dagger \m{H}_{ji,\mathrm{SISO}}^{[N]} \m{U}_{i,\mathrm{SISO}}}\nonumber\\
    &\quad
    + b\,\rank{\m{V}_{j,(2,1)}^\dagger \m{H}_{ji,(2,1)}^{[N]} \m{U}_{i,(2,1)}}\nonumber\\
    &\quad
    + c\,\rank{\m{V}_{j,(1,2)}^\dagger \m{H}_{ji,(1,2)}^{[N]} \m{U}_{i,(1,2)}} \nonumber\\
    &= 2a+3b+3c = K.
\end{align}
From \eqref{eq:null1} and \eqref{eq:null2}, it follows that
\begin{align}
    \begin{bmatrix}
        \m{V}_1^\dagger & -\m{V}_2^\dagger
    \end{bmatrix}
\end{align}
lies in the left null space of $\m{H}_U$. Hence there exists a matrix $\m{Q}$ such that
\begin{align}
    \begin{bmatrix}
        \m{V}_1^\dagger & -\m{V}_2^\dagger
    \end{bmatrix}
    =
    \m{Q}
    \bigl(
    \m{I}_{N(N_1+N_2)}
    -
    \m{H}_U(\m{H}_U^\dagger \m{H}_U)^{-1}\m{H}_U^\dagger
    \bigr),
\end{align}
because the row space of
\begin{align}
    \m{I}_{N(N_1+N_2)} - \m{H}_U(\m{H}_U^\dagger \m{H}_U)^{-1}\m{H}_U^\dagger
\end{align}
is precisely the full left null space of $\m{H}_U$.

Since we have explicitly constructed one assignment for which all determinants $f_{ji}$ are nonzero, and each $f_{ji}$ is a nonzero rational polynomial, it follows that for generic channel realizations, we have
\begin{align}
    \rank{\tilde{\m{H}}_{ji}} \overset{\text{a.s.}}{=} K.
\end{align}

Consequently, the two receivers observe identical $K\times K$ MACs with full-rank effective channels:
\begin{align}
    \tilde{\ve{y}}_1^{(K)}
    &= \tilde{\m{H}}_1 \tilde{\ve{x}}_1^{(K)} + \tilde{\m{H}}_2 \tilde{\ve{x}}_2^{(K)} + \tilde{\ve{z}}_1, \label{eq:MAC1} \\
    \tilde{\ve{y}}_2^{(K)}
    &= \tilde{\m{H}}_1 \tilde{\ve{x}}_1^{(K)} + \tilde{\m{H}}_2 \tilde{\ve{x}}_2^{(K)} + \tilde{\ve{z}}_2. \label{eq:MAC2}
\end{align}
The additive noise at each receiver has covariance of order $O(\sigma^2)$. Now let each transmitter further precode via $\tilde{\m{H}}_i^{-1}$, i.e.,
\begin{align}
    \tilde{\ve{x}}_i^{(K)} = \tilde{\m{H}}_i^{-1}\hat{\ve{x}}_i^{(K)}.
\end{align}
Then both effective MACs reduce to
\begin{align}
    \tilde{\ve{y}}_1^{(K)} &= \hat{\ve{x}}_1^{(K)} + \hat{\ve{x}}_2^{(K)} + \tilde{\ve{z}}_1, \\
    \tilde{\ve{y}}_2^{(K)} &= \hat{\ve{x}}_1^{(K)} + \hat{\ve{x}}_2^{(K)} + \tilde{\ve{z}}_2.
\end{align}
Thus, we obtain $K$ parallel scalar MACs, each with noise variance $O(\sigma^2)$ and transmit power independent of $\sigma^2$ (though dependent on the channel realization). By Lemma~\ref{lemma:SISOMAC}, each scalar MAC achieves $1$ AirComp DoF for both receivers. \new{Hence, a $(K,3)$ computation is feasible almost surely for a generic channel, yielding the achievable AirComp DoF}
\begin{align}
    \frac{K}{N}=\frac{2}{3}a+b+c.
\end{align}
This completes the proof.

\section{The MSE--Equivocation Bound}\label{appendix_mse}

\begin{lemma}[The MSE--Equivocation Bound]\label{lm:mse_equiv}
Given a $(K, N)$ computation code $\mathfrak{C}$ for an $(M_1, M_2; N_1, N_2; \mathscr{H})$ network, the conditional differential entropy $h(f^K |\mathbf{y}_j^N)$ for each receiver $j \in \{1,2\}$ is bounded by
\begin{align}
h(f^K |\mathbf{y}_j^N) \le K \log\left( \frac{\pi e}{K} \mathrm{MSE}(\mathfrak{C}, \mathscr{H}) \right).
\end{align}
Furthermore, if the computation is feasible, i.e., $\mathrm{MSE}(\mathfrak{C}, \mathscr{H}) = O(1/\rho)$ as $\rho \to \infty$, then
\begin{align} \label{eq:lm2.2_final}
h(f^K |\mathbf{y}_j^N) \le -K \log \rho + K o(\log \rho), \quad \rho \to \infty.
\end{align}
\end{lemma}

\begin{proof}
Fixing the code $\mathfrak{C}$ and realization $\mathscr{H}$ (suppressed from the conditioning for brevity), the estimate $\hat{f}_j^K$ is a deterministic function of $\mathbf{y}_j^N$, so
\begin{align}
h(f^K| \mathbf{y}_j^N)
&= h(f^K - \hat{f}_j^K | \mathbf{y}_j^N) \notag\\
&\stackrel{(a)}{\le} \sum_{k=1}^K h(f(k) - \hat{f}_j(k)| \mathbf{y}_j^N) \notag\\
&\stackrel{(b)}{\le} \sum_{k=1}^K h(f(k) - \hat{f}_j(k)) \notag\\
&\stackrel{(c)}{\le} \sum_{k=1}^K \log\left(\pi e \mathbb{E} [ | f(k) - \hat{f}_j(k)|^2 ]\right) \notag\\
&\stackrel{(d)}{\le} K \log\left(\frac{\pi e}{K} \sum_{k=1}^K \mathbb{E} [ | f(k) - \hat{f}_j(k)|^2 ]\right) \notag\\
&\stackrel{(e)}{\le} K \log\left(\frac{\pi e}{K} \mathrm{MSE}(\mathfrak{C}, \mathscr{H})\right),
\end{align}
where
    (a) the joint entropy is at most the sum of the marginals;
    (b) conditioning cannot increase entropy;
    (c) the Gaussian maximizes differential entropy for a fixed second moment;
    (d) the logarithm is concave; and (e) follows from the definition of $\mathrm{MSE}(\mathfrak{C}, \mathscr{H})$.
As $\rho \to \infty$, feasibility implies $\mathrm{MSE}(\mathfrak{C}, \mathscr{H}) \le c/\rho$. Substituting this into the RHS of (e)
completes the proof.
\end{proof}

\section{Proof of Invertibility Bound}\label{app:invert}

\begin{lemma}[Invertibility Bound]\label{lm:invert}
Let $\mathfrak{C}$ be a $(K,N)$ computation code for an $(M_1, M_2; N_1, N_2; \mathscr{H})$ network. If the channel matrices $\m{H}(\nu) \in \mathbb{C}^{(N_1 + N_2) \times (M_1 + M_2)}$ have full column rank for all $\nu \in [1:N]$, then
\begin{align}
    h(s_j^K|\ve{y}_1^N, \ve{y}_2^N)
    &\le -K \log \rho + K o(\log \rho) \notag\\
    &\quad + N o(\log \rho).
\end{align}
\end{lemma}
\new{\emph{Outline.} Since the stacked channel has full column rank, applying a left inverse to the received signals exposes the transmitted signals corrupted only by an effective, correlated Gaussian noise, \eqref{eq:inv_expose}. The main technical work is to show that this effective noise shifts the relevant differential entropy by at most $N o(\log\rho)$: the entropy is sandwiched between an elementary lower bound, \eqref{eq:wang_lower}, and the genie-aided upper bound of \cite{wangGenieChainsExploring2016}, \eqref{eq:wang_upper}, whose per-use gap is independent of the SNR, yielding the equality \eqref{eq:genie_final}. Combining this entropy comparison with the MSE--equivocation bound gives the stated $-K\log\rho$ source-equivocation bound in \eqref{eq:final_invert_bound}.}

\begin{proof}
Since $\m{H}(\nu)$ has full column rank, there exists a left inverse $\m{G}(\nu) \in \mathbb{C}^{(M_1 + M_2) \times (N_1 + N_2)}$ such that $\m{G}(\nu)\m{H}(\nu) = \m{I}$. Thus, applying this to the received signals:
\begin{align}
    \m{G}(\nu)
    \begin{bmatrix}
        \ve{y}_1(\nu) \\ \ve{y}_2(\nu)
    \end{bmatrix}
    =
    \begin{bmatrix}
        \ve{x}_1(\nu) \\
        \ve{x}_2(\nu)
    \end{bmatrix}
    +
    \m{G}(\nu)
    \begin{bmatrix}
        \ve{z}_1(\nu) \\
        \ve{z}_2(\nu)
    \end{bmatrix}. \label{eq:inv_expose}
\end{align}
The effective noise remains Gaussian; define
\begin{align}
    \begin{bmatrix}
       \tilde{\ve{z}}_1(\nu) \\ \tilde{\ve{z}}_2(\nu)
    \end{bmatrix}
    &\triangleq
    \m{G}(\nu)
    \begin{bmatrix}
       \ve{z}_1(\nu) \\ \ve{z}_2(\nu)
    \end{bmatrix}, \\
    \tilde{\ve{z}}_1(\nu) &\sim \mathcal{CN}(0, \sigma^2 \m{K}_1(\nu)), \\
    \tilde{\ve{z}}_2(\nu) &\sim \mathcal{CN}(0, \sigma^2 \m{K}_2(\nu)).
\end{align}
Since $\m{G}(\nu)$ has full row rank, $\m{K}_1(\nu)$ and $\m{K}_2(\nu)$ are positive definite. Focusing on $s_1^K$, we bound the equivocation:
\begin{align}
    h(s_1^K&|\ve{y}_1^N,\ve{y}_2^N) \notag \\
    &= h(s_1^K|\ve{y}_1^N,\ve{y}_2^N, (\ve{x}_1 + \tilde{\ve{z}}_1)^N, (\ve{x}_2 + \tilde{\ve{z}}_2)^N) \label{eq:inv_1} \\
    &\stackrel{(a)}{\le} h(s_1^K|(\ve{x}_1 + \tilde{\ve{z}}_1)^N) \label{eq:inv_2} \\
    &\stackrel{(b)}{=} h(s_1^K|(\ve{x}_1 + \tilde{\ve{z}}_1)^N, s_2^K, \ve{z}_1^{\prime N}) \label{eq:inv_3} \\
    &= h(f^K|(\m{H}_{11} \ve{x}_1 + \m{H}_{12} \ve{x}_2 +  \m{H}_{11} \tilde{\ve{z}}_1 +\ve{z}^\prime_1)^N, \notag \\
    &\quad (\ve{x}_1 + \tilde{\ve{z}}_1)^N, s_2^K, \ve{z}_1^{\prime N}) \label{eq:inv_4} \\
    &\stackrel{(c)}{\le} h(f^K|(\m{H}_{11} \ve{x}_1 + \m{H}_{12} \ve{x}_2 +  \m{H}_{11} \tilde{\ve{z}}_1 +\ve{z}^\prime_1)^N). \label{eq:inv_5}
\end{align}
where
    (a) follows because conditioning cannot increase entropy;
(b) follows because $s_2^K$ and the auxiliary random sequence $\ve{z}_1^{\prime N} \sim \mathcal{CN}(0, \sigma^2 \m{I})$ are independent of $s_1^K$ and $(\ve{x}_1 + \tilde{\ve{z}}_1)^N$; and
(c) follows because conditioning reduces entropy.

Next, we bound the term in \eqref{eq:inv_5} using the genie-aided bound of \cite{wangGenieChainsExploring2016} (equation (230) therein):
\begin{align}
    h\big((\m{H}_{11} \ve{x}_1 &+ \m{H}_{12} \ve{x}_2 + \ve{z}^\prime_1)^N\big)  \label{eq:wang_lower}\\
    &\le h\big((\m{H}_{11} \ve{x}_1 + \m{H}_{12} \ve{x}_2 + \m{H}_{11} \tilde{\ve{z}}_1 + \ve{z}^\prime_1)^N\big) \label{eq:entropy_lower_bound} \\
    &\stackrel{(d)}{\le} h\big((\m{H}_{11} \ve{x}_1 + \m{H}_{12} \ve{x}_2 + \ve{z}^\prime_1)^N\big) \notag \\
    &\quad + \log \det(\tilde{\m{K}} + \m{I}). \label{eq:wang_upper}
\end{align}
where (d) applies the bound from \cite{wangGenieChainsExploring2016}, and the covariance matrix of the additional noise is defined as:
\begin{align} \label{eq:cov_matrix_def}
    \tilde{\m{K}} &= \E{\m{H}_{11}^{[N]} \tilde{\ve{z}}_1^{(N)} \tilde{\ve{z}}_1^{(N) \dagger} \m{H}_{11}^{[N] \dagger}}  \E{\ve{z}_1^{\prime (N)} \ve{z}_1^{\prime (N)\dagger}}^{-1}  \\
    &= \m{H}_{11}^{[N]} \m{K}_1^{[N]} \m{H}_{11}^{[N]\dagger}.
\end{align}
Because $\tilde{\m{K}}$ is block diagonal across channel uses, the determinant factors as
\begin{align}
    \log \det(\tilde{\m{K}} + \m{I}) = \sum_{\nu =1}^{N} \log \det(\m{H}_{11}(\nu) \m{K}_1(\nu) \m{H}_{11}(\nu)^\dagger + \m{I}).
\end{align}
As $\rho \to \infty$, each term in the sum is $o(\log \rho)$, so \eqref{eq:wang_lower} and \eqref{eq:wang_upper} give
\begin{align} \label{eq:genie_final}
    h\big(&(\m{H}_{11} \ve{x}_1 + \m{H}_{12} \ve{x}_2 + \m{H}_{11} \tilde{\ve{z}}_1 +\ve{z}^\prime_1)^N\big) \notag \\
    &= h\big((\m{H}_{11} \ve{x}_1 + \m{H}_{12} \ve{x}_2 + \ve{z}^\prime_1)^N\big) + N o(\log \rho).
\end{align}
Substituting \eqref{eq:genie_final} into \eqref{eq:inv_5}, and noting that $(\m{H}_{11} \ve{x}_1 + \m{H}_{12} \ve{x}_2 +  \ve{z}^\prime_1)^N$ has the same distribution as $\ve{y}_1^N$:
\begin{align}
    h(s_1^K|\ve{y}_1^N,\ve{y}_2^N)
    &\le h(f^K|\ve{y}_1^N) + N o(\log \rho) \label{eq:inv_6} \\
    &\stackrel{(e)}{\le} - K \log \rho
    + K o(\log \rho) + N o(\log \rho), \label{eq:final_invert_bound}
\end{align}
where (e) follows directly from the MSE--equivocation bound.
\end{proof}

\section{Power Bound}\label{appendix_4}

\begin{lemma}[Power Bound]\label{lm:powerbound}
The differential entropy of the signal received at Rx-$j$, as $\rho\rightarrow\infty$, is bounded by
\begin{align}
  h(\ve{y}_j^N) \le NN_j \log \left((M_1+M_2)M_o^2  + \frac{1}{\rho} \right) + N o(\log \rho).
\end{align}
\end{lemma}

\begin{proof}
Without loss of generality the code generates zero-mean signals (a nonzero mean consumes power without carrying information). The total average received power across $N$ channel uses is:
\begin{align}
    &\mathbb{E} \{ \|\ve{y}_j^{(N)} \|^2 \}\notag\\
    &= \sum_{n=1}^N \mathbb{E}\Big\|
    \m{H}_{j1}(n) \ve{x}_1(n)
    + \m{H}_{j2}(n) \ve{x}_2(n)
    + \ve{z}_j(n) \Big\|^2 \label{eq:pwr_1} \\
    &\stackrel{(a)}{=} \sum_{n=1}^N \Big(
    \mathbb{E} \|\m{H}_{j1}(n) \ve{x}_1(n)\|^2 \notag\\
    &\qquad\qquad + \mathbb{E} \|\m{H}_{j2}(n) \ve{x}_2(n)\|^2
    + \mathbb{E} \|\ve{z}_j(n)\|^2 \Big) \label{eq:pwr_2} \\
    &\stackrel{(b)}{\le} \sum_{n=1}^N \Big( \|\m{H}_{j1}(n)\|_F^2
    \mathbb{E} \|\ve{x}_1(n)\|^2 \notag\\
    &\qquad\qquad + \|\m{H}_{j2}(n)\|_F^2
    \mathbb{E} \|\ve{x}_2(n)\|^2 \Big)
    + N N_j \sigma^2 \label{eq:pwr_3} \\
    &\stackrel{(c)}{\le} N N_j (M_1 + M_2)M_o^2 + NN_j \sigma^2. \label{eq:pwr_4}
\end{align}
where (a) follows because $\ve{x}_1$, $\ve{x}_2$, and $\ve{z}_j$ are independent and zero-mean, so all cross-terms vanish; (b) follows from the submultiplicative property of the Frobenius norm ($\|\m{H}\ve{x}\| \le \|\m{H}\|_F \|\ve{x}\|$); and (c) follows from the normalized transmit power constraint \eqref{eq:P} and the fact that the squared magnitude of each element in $\m{H}_{ji}(n)$ is bounded by $M_o^2$, yielding $\|\m{H}_{ji}(n)\|_F^2 \le N_j M_i M_o^2$.

Since a circularly symmetric complex Gaussian maximizes differential entropy under a power constraint, the output entropy is bounded by
\begin{align}
    h(\ve{y}_j^N) &\stackrel{(d)}{\le} \log \det\left( \pi e \frac{\E{\lVert \ve{y}_j^{(N)} \rVert^2}}{NN_j} \m{I}_{NN_j} \right) \label{eq:pwr_5} \\
    &\stackrel{(e)}{\le} \log \det\left(\pi e ((M_1 + M_2) M_o^2 +  \sigma^2)\m{I}_{NN_j}\right) \label{eq:pwr_6} \\
    &= NN_j \log \left(\pi e ((M_1 + M_2) M_o^2 +  \sigma^2)\right) \label{eq:pwr_7} \\
    &= NN_j \log \left((M_1 + M_2)M_o^2+ \frac{1}{\rho}\right) + NN_j \log(\pi e). \label{eq:pwr_8}
\end{align}
Here (d) applies the maximum entropy theorem for a complex random vector of dimension $NN_j$, and (e) substitutes \eqref{eq:pwr_4}; since $NN_j \log(\pi e)=N o(\log \rho)$, the proof is complete.
\end{proof}

\bibliographystyle{IEEEtran}
\bibliography{ref}

\end{document}